\documentclass{llncs}

\pdfoutput=1

\usepackage{amsmath}
\usepackage{amssymb}
\usepackage{changepage}
\usepackage[none]{hyphenat}
\usepackage[inline]{enumitem}
\usepackage[nomargin,inline,final]{fixme}
\usepackage{graphicx}
\usepackage{hyperref}
\usepackage{mathtools}
\usepackage{microtype}
\usepackage{placeins}
\usepackage{proof}
\usepackage{prooftree}
\usepackage{relsize}
\usepackage{rotating}
\usepackage{subfig}
\usepackage{textcomp}
\usepackage{thm-restate}
\usepackage{tikz}
\usepackage{url}
\usepackage{versions}

\usepackage[capitalise]{cleveref}

\crefname{rule}{Rule}{Rules}
\creflabelformat{rule}{(#2#1#3)}
\crefrangelabelformat{rule}{(#3#1#4) to~(#5#2#6)}

\usetikzlibrary{arrows}
\usetikzlibrary{shapes.arrows}
\usetikzlibrary{calc}
\usetikzlibrary{positioning}

\tikzstyle{arrow}+=[thick,rounded corners=0.5em]
\tikzstyle{every picture}+=[remember picture,baseline]

\fxusetheme{color}
\fxuseenvlayout{color}

\FXRegisterAuthor{rr}{arr}{\color{teal}[reuben]}
\FXRegisterAuthor{lc}{alc}{\color{magenta}[liron]}
\FXRegisterAuthor{cut}{acut}{\color{orange}[cut]}

\include{rtc-macros}

\makeatletter
\renewcommand\subsubsection{\@startsection{subsubsection}{3}{\z@}%
                       {-18\p@ \@plus -4\p@ \@minus -4\p@}%
                       {0.5em \@plus 0.22em \@minus 0.1em}%
                       {\normalfont\normalsize\bfseries\boldmath}}
\makeatother
\begin{document}

\title{Infinitary and Cyclic Proof Systems for Transitive Closure Logic}
\author{Liron Cohen\inst{1} \and Reuben N.~S.~Rowe\inst{2}}
\institute{
  Dept.~of Computer Science, Cornell University, NY, USA, \email{lironcohen@cornell.edu}
    \and
  School of Computing, University of Kent, Canterbury, UK, \email{r.n.s.rowe@kent.ac.uk}
}

\maketitle

\begin{abstract}
Transitive closure logic is a known extension of first-order logic obtained by introducing a transitive closure operator. While other extensions of first-order logic with inductive definitions are a priori parametrized by a set of inductive definitions, the addition of the transitive closure operator uniformly captures all finitary inductive definitions.
In this paper we present an \emph{infinitary} proof system for transitive closure logic which is an \emph{infinite descent}-style counterpart to the existing (explicit induction) proof system for the logic. We show that, as for similar systems for first-order logic with inductive definitions, our infinitary system is complete for the standard semantics and subsumes the explicit system. Moreover, the uniformity of the transitive closure operator allows semantically meaningful complete restrictions to be defined using simple syntactic criteria. Consequently, the restriction to regular infinitary (i.e.~\emph{cyclic}) proofs provides the basis for an effective system for automating inductive reasoning.
\end{abstract}

\section{Introduction}

A core technique in mathematical reasoning is that of \emph{induction}. This is especially true in computer science, where it plays a central role in reasoning about recursive data and computations. Formal systems for mathematical reasoning usually capture the notion of inductive reasoning via one or more inference rules that express the general induction schemes, or principles, that hold for the elements being reasoned over. 
%For example, when reasoning over natural numbers (i.e.~the business of arithmetic), one first observes that these are `generated' from a `base' element $0$ and the successor function $s$. This then provides a sufficient condition for proving that all numbers satisfy a given property: show that it is both satisfied by the base element and preserved in the generation of new elements. In a similar manner, different schemes for generating other collections of elements give rise to alternative schemes of inductive reasoning, applicable to those collections.

Increasingly, we are concerned with not only being able to formalise as much mathematical reasoning as possible, but also with doing so in an effective way. In other words, we seek to be able to automate such reasoning. \emph{Transitive closure} ({\TC}) logic has been identified as a potential candidate for a minimal, `most general' system for inductive reasoning, which is also very suitable for automation \cite{AvronTC03,Cohen2014AL,cohen2015middle}. {\TC} adds to first-order logic a single  operator for forming binary relations: specifically, the transitive closures of arbitrary formulas (more precisely, the transitive closure of the binary relation induced by a formula with respect to two distinct variables). 
% Do we need the following sentence?
In this work, for simplicity, we use a reflexive form of the operator; however the two forms are equivalent in the presence of equality. 
This modest addition affords enormous expressive power: namely it provides a uniform way of capturing inductive principles. If an induction scheme is expressed by a formula $\varphi$, then the elements of the inductive collection it defines are those `reachable' from the base elements $x$ via the iteration of the induction scheme. That is, those $y$'s for which $(x, y)$ is in the transitive closure of $\varphi$. Thus, bespoke induction principles do not need to be added to, or embedded within, the logic; instead, all induction schemes are available within a single, unified language. In this respect, the transitive closure operator resembles the W-type \cite{martin1984intuitionistic}, which also provides a single type constructor from which one can uniformly define a variety of inductive types.

{\TC} logic is intermediate between first- and second-order logic. Furthermore, since the {\TC} operator is a particular instance of a least fixed point operator, {\TC} logic is also subsumed by fixed-point logics such as the $\mu$-calculus \cite{kashima2008general}. %
\lcnote{also cite Kozen. Results on the propositional mu-calculus.}%
However, despite its minimality {\TC} logic retains enough expressivity to capture inductive reasoning, as well as to subsume arithmetics (see \cref{arithmetics}). Moreover, from a proof theoretical perspective the conciseness of the logic makes it of particular interest. The use of only one constructor of course comes with a price: namely, formalizations (mostly of non-linear induction schemes) may be somewhat complex. However, they generally do not require as complex an encoding as in arithmetics, since the {\TC} operator can be applied on any formula and thus (depending on the underlying signature) more naturally encode induction on sets more complex than the natural numbers. %
\lcnote{do you think the par is helpful or is it too vague? }%
\rrnote{I think this is fine - I have tweaked it a little. }%

Since its expressiveness entails that {\TC} logic subsumes arithmetics, by G\"{o}del's result, any effective proof system for it must necessarily be incomplete for the standard semantics. Notwithstanding, a natural, effective proof system which is sound for {\TC} logic was shown to be complete with respect to a generalized form of Henkin semantics \cite{Cohen2017Henkin}.
In this paper, following similar developments in other formalizations for fixed point logics and inductive reasoning (see e.g.~\cite{Brotherston07,BrotherstonBC08,brotherston2010sequent,Santocanale2002,Sprenger2003}), we present an \emph{infinitary} proof theory for {\TC} logic which, as far as we know, is the first system that is (cut-free) complete with respect to the standard semantics.
More specifically, our system employs infinite-height, rather than infinite-width proofs (see \cref{sec:ProofSystems:Infinitary}).
The soundness of such infinitary proof theories is underpinned by the principle of \emph{infinite descent}:
% \cite{Brotherston07,BrotherstonBC08,brotherston2010sequent,DasP17,RoweBrotherston17,TellezBrotherston17}.
proofs are permitted to be infinite, non-well-founded trees, but subject to the restriction that every infinite path in the proof admits some infinite descent.
The descent is witnessed by tracing terms or formulas for which we can give a correspondence with elements of a well-founded set.
In particular, we can trace terms that denote elements of an inductively defined (well-founded) set.
%In the context of formalized induction, we can use formulas interpreted by the elements of inductive collections.
For this reason, such theories are considered systems of implicit induction, as opposed to those which employ explicit rules for applying induction principles.
While a full infinitary proof theory is clearly not effective, in the aforementioned sense, such a system can be obtained by restricting consideration to only the \emph{regular} infinite proofs. These are precisely those proofs that can be finitely represented as (possibly cyclic) graphs.

These infinitary proof theories generally subsume systems of explicit induction in expressive power, but also offer a number of advantages. 
Most notably, they can ameliorate the primary challenge for inductive reasoning: finding an induction \emph{invariant}. 
In explicit induction systems, this must be provided \emph{a priori}, and is often much stronger than the goal one is ultimately interested in proving.
However, in implicit systems the inductive arguments and hypotheses may be encoded in the cycles of a proof, so cyclic proof systems seem better for automation.
The cyclic approach has also been used to provide an optimal cut-free complete proof system for Kleene algebra \cite{DasP17}, providing further evidence of its utility for automation.

In the setting of {\TC} logic, we observe some further benefits over more traditional formal systems of inductive definitions and their infinitary proof theories (cf.~{\lkid} \cite{brotherston2010sequent,Martin-Lof71}). {\TC} (with a pairing function) has all first-order definable finitary inductive definitions immediately `available' within the language of the logic: as with inductive hypotheses, one does not need to `know' in advance which induction schemes will be required. Moreover, the use of a single transitive closure operator provides a uniform treatment of all induction schemes. That is, instead of having a proof system parameterized by a set of inductive predicates and rules for them (as is the case in {\lkid}), {\TC} offers a single proof system with a single rule scheme for induction. This has immediate advantages for developing the metatheory: the proofs of completeness for standard semantics and adequacy (i.e.~subsumption of explicit induction) for the infinitary system presented in this paper are simpler and more straightforward. %, as is the proof of equivalence for the two systems under arithmetic\lcnote{we still need to sort this out}.
Moreover, it permits a cyclic subsystem, which also subsumes explicit induction, to be defined via a simple syntactic criterion that we call \emph{normality}. The smaller search space of possible proofs further enhances the potential for automation.
{\TC} logic seems more expressive in other ways, too. For instance, the transitive closure operator may be applied to arbitrarily complex formulas, %thus we are not restricted to induction principles corresponding only to \emph{monotone} generation schemes 
not only to collections of atomic formulas (cf.~Horn clauses), as in e.g. \cite{Brotherston07,brotherston2010sequent}.
%We explore this aspect of {\TC} logic in \cref{sec:lkid} by comparing the expressive power of {\TC} logic with that of {\lkid}, a system of first order Martin-L\"{o}f style inductive definitions.\lcnote{we still have to do this} We also believe the uniformity provided by the transitive closure operator means that we can incorporate \emph{co}inductive reasoning without any additional extensions, although this remains future work.\lcnote{this is another `believe' sentence. Though I would really like to hint at co-induction in the intro, I feel that if we do not have anything to initial about it in the body of the paper, this is more suitable for future work sec.}

We show that the explicit and cyclic {\TC} systems are equivalent under arithmetic, as is the case for {\lkid} \cite{Berardi2017b,Simpson2017}. However, there are cases in which the cyclic system for {\lkid} is strictly more expressive than the explicit induction system \cite{Berardi2017}. To obtain a similar result for {\TC}, the fact that all induction schemes are available poses a serious challenge. For one,  the counter-example used in \cite{Berardi2017}  does not serve to show this result holds for {\TC}. If this strong inequivalence indeed holds also for {\TC}, it must be witnessed by a more subtle and complex counter-example. Conversely, it may be that the explicit and cyclic systems do coincide for {\TC}. In either case, this points towards fundamental aspects that require further investigation.

The rest of the paper is organised as follows. In \cref{sec:TCLogicDef} we reprise the definition of transitive closure logic and both its standard and Henkin-style semantics. 
\Cref{sec:ProofSystems} presents the existing explicit induction proof system for {\TC} logic, and also our new infinitary proof system. 
%While the former is sound and complete for the Henkin semantics, 
We prove the latter sound and complete for the standard semantics, and also derive cut-admissibility. 
In \cref{sec:ProofSystemComparisons} we compare the expressive power of the infinitary system (and its cyclic subsystem) with the explicit system. 
%In \cref{sec:lkid} we consider the relationship between our systems for {\TC} logic and the explicit and infinitary systems for systems of inductive definitions for first order logic.
\Cref{sec:FutureWork} concludes and examines the remaining open questions for our system as well as future work. 

This technical report comprises an extended version, with proofs, of the results presented in \cite{CohenRowe18}. We would like to thank an anonymous reviewer for bringing to our attention a technical problem with the proof of completeness for the infinitary system in a previous version of this work.

\section{Transitive Closure Logic and its Semantics}
\label{sec:TCLogicDef}

In this section we review the language of transitive closure logic, and two possible semantics for it: a standard one, and a Henkin-style one. 
For simplicity of presentation we assume (as is standard practice) a designated equality symbol in the language.
We denote by $v[x_1 := a_n, \ldots, x_n := a_n]$ the variant of the assignment $v$ which assigns $a_i$ to $x_i$ for each $i$, and by $\varphi \subst{\replace{t_{1}}{x_{1}}, \ldots,\replace{t_{n}}{x_{n}}} $ the result of simultaneously substituting each $t_{i}$ for the free occurrences of $x_{i}$ in $\varphi$.

\begin{definition}[The language \rtclogic]
\label[definition]{def:TC}
  Let $\sigma$ be a first-order signature with equality, whose terms are ranged over by $s$ and $t$ and predicates by $P$, and let $x$, $y$, $z$, etc.~range over a countable set of variables. The language {\rtclogic} consists of the formulas defined by the grammar:
  \begin{align*}
    \varphi, \psi \Coloneqq {}
         & s=t
      \mid P(t_1, \ldots, t_n) 
      \mid \neg \varphi 
      \mid \varphi \wedge \varphi 
      \mid \varphi \vee \varphi 
      \mid \varphi \rightarrow \varphi 
      \mid \forall x . \varphi 
      \mid \exists x . \varphi
      \mid 
        \\
        & (\rtcformula{x}{y}{\varphi})(s,t)
  \end{align*}
  As usual, $\forall x$ and $\exists x$ bind free occurrences of the variable $x$ and we identify formulas up to renaming of bound variables, so that capturing of free variables during substitution does not occur. Note that in the formula $(\rtcformula{x}{y}{\varphi})(s,t)$ free occurrences of $x$ and $y$ in $\varphi$ are also bound (but not those in $s$ and $t$).
\end{definition}

\begin{definition}[Standard Semantics]
  \label[definition]{def:Semantics:Standard}
  Let $M = \langle D, I \rangle$ be a first-order structure (i.e.~$D$ is a non-empty domain and $I$ an interpretation function), and $v$ an assignment in $M$ which we extend to terms in the obvious way. 
  The satisfaction relation $\models$ between model-valuation pairs $\langle M, v \rangle$ and formulas is defined inductively on the structure of formulas by:
%  We define the semantics of {\rtclogic} in the usual way for the standard first-order components, and define the semantics of formulas of the form $(\rtcformula{x}{y}{\varphi})(s,t)$ as follows:
  \begin{itemize}
    \item $M, v \models s=t$ if $v(s)=v(t)$;
    \item $M, v \models P(t_1, \ldots, t_n)$ if $(v(t_1), \ldots, v(t_n)) \in I(P)$; 
    \item $M, v \models \neg \varphi$ if $M, v \not\models \varphi$;
    \item $M, v \models \varphi_1 \wedge \varphi_2$ if both $M, v \models \varphi_1$ and $M, v \models \varphi_2$;
    \item $M, v \models \varphi_1 \vee \varphi_2$ if either $M, v \models \varphi_1$ or $M, v \models \varphi_2$;
    \item $M, v \models \varphi_1 \rightarrow \varphi_2$ if $M, v \models \varphi_1$ implies $M, v \models \varphi_2$;
    \item $M, v \models \exists x . \varphi$ if $M, v[x := a] \models \varphi$ for some $a \in D$;
    \item $M, v \models \forall x . \varphi$ if $M, v[x := a] \models \varphi$ for all $a \in D$;
    \item
    $M, v \models (\rtcformula{x}{y}{\varphi})(s,t)$ if $v(s) = v(t)$, or there exist $a_{0}, \ldots, a_{n} \in D$ ($n > 0$) s.t. $v(s) = a_{0}$, $v(t) = a_{n}$, and $M, v[x := a_{i}, y := a_{i+1}] \models \varphi$ for $0 \leq i < n$.
  \end{itemize}
  We say that a formula $\varphi$ is valid with respect to the standard semantics when $M, v \models \varphi$ holds for all models $M$ and valuations $v$.
\end{definition}

We next recall the concepts of frames and Henkin structures (see, e.g., \cite{Henkin50Completeness}). A frame is a first-order structure together with some subset of the powerset of its domain (called its set of admissible subsets).

\begin{definition}[Frames] 
  A frame $M$ is a triple $\langle D, I, \mathcal{D} \rangle$, where $\langle D, I \rangle$ is a first-order structure, and $\mathcal{D} \subseteq \wp(D)$. 
\end{definition}
Note that if $\mathcal{D} = \wp(D)$, the frame is identified with a standard first-order structure.

\begin{definition}[Frame Semantics]
%  Let {\rtclogic} be the language based on the first-order signature $\sigma$. 
  {\rtclogic} formulas are interpreted in frames as in \cref{def:Semantics:Standard} above, except for:% the following clause concerning the $\rtcOperator$ operator:
  \begin{itemize}
    \item
    $M, v \models (\rtcformula{x}{y}{\varphi})(s,t)$ if for every $A \in \mathcal{D}$, if $v(s) \in A$ and for every $a, b \in D$: $a \in A$ and $M, v[x := a, y := b] \models \varphi$ implies $b \in A$, then $v(t) \in A$.
  \end{itemize}
\end{definition}

We now consider Henkin structures, which are frames whose set of admissible subsets is closed under parametric definability.

\begin{definition}[Henkin structures] 
  A Henkin structure $M = \langle D, I, \mathcal{D} \rangle$ is a frame such that $\{ a \in D \,\mid\, M, v[x := a] \models \varphi \} \in \mathcal{D}$ for every $\varphi$, and $v$ in $M$.
\end{definition}
We refer to the semantics induced by quantifying over the (larger) class of Henkin structures as the Henkin semantics.

It is worth noting that the inclusion of equality in the basic language is merely for notational convenience. This is because the {\rtcOperator} operator allows us, under both the standard and Henkin semantics, to actually \emph{define} equality $s = t $ on terms as $(\rtcformula{x}{y}{\bot})(s,t)$.

\section{Proof Systems for \rtclogic}
\label{sec:ProofSystems}

In this section, we define two proof systems for {\rtclogic}. The first is a finitary proof system with an explicit induction rule for {\rtcOperator} formulas. The second is an infinitary proof system, in which {\rtcOperator} formulas are simply unfolded, and inductive arguments are represented via infinite descent-style constructions. We show the soundness and completeness of these proof systems, and also compare their provability relations.

\begin{figure}[t]
%  \smaller
%  \begin{adjustwidth}{-1em}{}
  \begin{gather*}
    \text{(Axiom):} \;
    \begin{prooftree}
      \phantom{\emptyset}
      \justifies
      \sequent{\varphi}{\varphi}
    \end{prooftree}
%      \qquad
%    \text{($\top$)} \;
%    \begin{prooftree}
%      \phantom{\emptyset}
%      \justifies
%      \sequent{}{\top}
%    \end{prooftree}
%      \qquad
%    \text{($\bot$)} \;
%    \begin{prooftree}
%      \phantom{\emptyset}
%      \justifies
%      \sequent{\bot}{}
%    \end{prooftree}
      \qquad
    \text{(WL):} \;
    \begin{prooftree}
      \sequent{\Gamma}{\Delta}
      \justifies
      \sequent{\Gamma, \varphi}{\Delta}
    \end{prooftree}    
      \qquad \;
    \text{(WR):} \;
    \begin{prooftree}
      \sequent{\Gamma}{\Delta}
      \justifies
      \sequent{\Gamma}{\Delta, \varphi}
    \end{prooftree}
      \\[1em]
    \text{($\vee$L):} \;
    \begin{prooftree}
      \sequent{\Gamma}{\varphi, \Delta}
      \quad
      \sequent{\Gamma, \psi}{\Delta}
      \justifies
      \sequent{\Gamma, \varphi \vee \psi}{\Delta}
    \end{prooftree}
       \quad
           \text{($\wedge$L):} \;
    \begin{prooftree}
      \sequent{\Gamma, \varphi, \psi}{\Delta}
      \justifies
      \sequent{\Gamma, \varphi \wedge \psi}{\Delta}
    \end{prooftree}
      \quad
    \text{($\rightarrow$L):} \;
    \begin{prooftree}
      \sequent{\Gamma}{\varphi, \Delta}
      \quad
      \sequent{\Gamma, \psi}{\Delta}
      \justifies
      \sequent{\Gamma, \varphi \rightarrow \psi}{\Delta}
    \end{prooftree}
      \\[1em]
    \text{($\vee$R):} \;
    \begin{prooftree}
      \sequent{\Gamma}{\varphi, \psi, \Delta}
      \justifies
      \sequent{\Gamma}{\varphi \vee \psi, \Delta}
    \end{prooftree}
     \quad
         \text{($\wedge$R):} \;
    \begin{prooftree}
      \sequent{\Gamma}{\varphi, \Delta}
      \quad
      \sequent{\Gamma}{\psi, \Delta}
      \justifies
      \sequent{\Gamma}{\varphi \wedge \psi, \Delta}
    \end{prooftree}
      \quad
    \text{($\rightarrow$R):} \;
    \begin{prooftree}
      \sequent{\Gamma, \varphi}{\psi, \Delta}
      \justifies
      \sequent{\Gamma}{\varphi \rightarrow \psi, \Delta}
    \end{prooftree}
      \\[1em]
    \text{($\exists$L):} \;
    \begin{prooftree}
      \sequent{\Gamma, \varphi}{\Delta}
      \justifies
      \sequent{\Gamma, \exists x . \varphi}{\Delta}
      \using x \not\in \fn{fv}(\Gamma, \Delta)
    \end{prooftree}
      \quad     
    \text{($\forall$L):} \;
    \begin{gathered}
    \infer
      {\sequent{\Gamma, \forall x . \varphi}{\Delta}}
      {\sequent{\Gamma, \varphi \subst{\replace{t}{x}}}{\Delta}}
    \end{gathered}
      \quad
    \text{($\neg$L):} \;
    \begin{prooftree}
      \sequent{\Gamma}{\varphi, \Delta}
      \justifies
      \sequent{\Gamma, \neg \varphi}{\Delta}
    \end{prooftree}
      \\[0.5em]
    \text{($\exists$R):} \;
    \begin{gathered}
    \infer
      {\sequent{\Gamma}{\exists x . \varphi, \Delta}}
      {\sequent{\Gamma}{\varphi\subst{\replace{t}{x}}, \Delta}}
    \end{gathered}
      \quad
    \text{($\forall$R):} \;
    \begin{prooftree}
      \sequent{\Gamma}{\varphi, \Delta}
      \justifies
      \sequent{\Gamma}{\forall x . \varphi, \Delta}
      \using x \not\in \fn{fv}(\Gamma, \Delta)
    \end{prooftree}
      \quad
    \text{($\neg$R):} \;
    \begin{prooftree}
      \sequent{\Gamma, \varphi}{\Delta}
      \justifies
      \sequent{\Gamma}{\neg \varphi, \Delta}
    \end{prooftree}
      \\[0.5em]
    \text{(${=}\text{L}_1$):} \;
    \begin{gathered}
    \infer
      {\sequent{\Gamma, s = t}{\varphi \subst{\replace{t}{x}}, \Delta}}
      {\sequent{\Gamma}{\varphi \subst{\replace{s}{x}}, \Delta}}
   \end{gathered}
      \quad
    \text{(${=}\text{L}_2$):} \;
    \begin{gathered}
    \infer
      {\sequent{\Gamma, s = t}{\varphi \subst{\replace{s}{x}}, \Delta}}
      {\sequent{\Gamma}{\varphi \subst{\replace{t}{x}}, \Delta}}
   \end{gathered}
     \quad
    \text{($=$R):} \;
    \begin{prooftree}
      \phantom{\sequent{}{t=t}}
      \justifies
      \sequent{}{t=t}
    \end{prooftree}
      \\[0.5em]
    \text{(Cut):} \;
    \begin{prooftree}
      \sequent{\Gamma}{\varphi, \Delta}
      \quad
      \sequent{\Sigma, \varphi}{\Pi}
      \justifies
      \sequent{\Gamma, \Sigma}{\Delta, \Pi}
    \end{prooftree}
      \quad
    \text{(Subst):} \;
    \begin{gathered}
%    \vspace{1em}
    \infer
      {\sequent
        {\Gamma \subst{\replace{t_1}{x_1}, \ldots, \replace{t_n}{x_n}}}
        {\Delta \subst{\replace{t_1}{x_1}, \ldots, \replace{t_n}{x_n}}}}
      {\sequent{\Gamma}{\Delta}{\mathrlap{\parbox[b][1.6em][c]{0pt}{\phantom{\subst{\replace{t_1}{x_1}, \ldots, \replace{t_n}{x_n}}}}}}}
    \end{gathered}
  \end{gather*}
%  \end{adjustwidth}
  \caption{Proof rules for the sequent calculus $\mathcal{LK}_=$ with substitution.}
  \label{fig:LKProofRules}
\end{figure}

Our systems for $\rtclogic$ are extensions of $\mathcal{LK}_=$, the  sequent calculus for classical first-order logic with equality  \cite{Gentzen1935,takeuti2013proof} whose proof rules we show in \cref{fig:LKProofRules}.\footnote{Here we  take $\mathcal{LK}_=$ to include the substitution rule, which was not a part of the original systems.}
Sequents are expressions of the form $\sequent{\Gamma}{\Delta}$, for finite sets of formulas $\Gamma$ and $\Delta$. We write $\Gamma, \Delta$ and $\Gamma, \varphi$ as a shorthand for $\Gamma \cup \Delta$ and $\Gamma \cup \{ \varphi \}$ respectively, and $\fn{fv}(\Gamma)$ for the set of free variables of the formulas in the set $\Gamma$. A sequent $\sequent{\Gamma}{\Delta}$ is valid if and only if the formula $\bigwedge_{\varphi \in \Gamma} \varphi \rightarrow \bigvee_{\psi \in \Delta} \psi$ is.

\subsection{The Finitary Proof System}
\label{sec:ProofSystems:Finitary}

We briefly summarise the finitary proof system for {\rtclogic}. For more details see \cite{Cohen2014AL,cohen2015middle}.
We write $\varphi(x_1, \ldots, x_n)$ to emphasise that the formula $\varphi$ may contain $x_1, \ldots, x_n$ as free variables.

\begin{definition}
  The proof system {\rtc} for {\rtclogic} is defined by adding to $\LK_=$ %(with equality) 
  the following inference rules:
  \begin{gather}
    \label[rule]{eq:refRTC}
    \begin{gathered}
    \infer{\sequent{\Gamma}{\Delta, (\rtcformula{x}{y}{\varphi})(s, s)}}{}
    \end{gathered}
      \\[0.5em]
    \label[rule]{eq:transRTC}
    \begin{gathered}
    \infer
    {\sequent{\Gamma}{\Delta, (\rtcformula{x}{y}{\varphi})(s, t)}}
    { \sequent{\Gamma}{\Delta, (\rtcformula{x}{y}{\varphi})(s, r)}
        &
      \sequent{\Gamma}{\Delta, {\varphi \subst{\scalebox{0.8}{\replace{r}{x}, \replace{t}{y}}}}}
    }
    \end{gathered}
      \\[0.5em]
    \label[rule]{eq:RTC-min}
    \begin{gathered}
    \infer
    {\sequent
      {\Gamma, {\psi \subst{\replace{s}{x}}}, (\rtcformula{x}{y}{\varphi})(s, t)}
      {\Delta, {\psi \subst{\replace{t}{x}}}}}
    {\sequent
      {\Gamma, \psi(x), \varphi(x,y)}
      {\Delta, {\psi \subst{\replace{y}{x}}}}}
    \end{gathered}
  \end{gather}
  where, for \cref{eq:RTC-min}, $x \not\in \fn{fv}(\Gamma, \Delta)$ and $y \not\in \fn{fv}(\Gamma, \Delta, \psi)$
\end{definition}

\Cref{eq:RTC-min} is a generalized induction principle. It states that if an extension of formula $\psi$ is closed under the relation induced by $\varphi$, then it is also closed under the reflexive transitive closure of that relation. In the case of arithmetic this rule captures the induction rule of Peano's Arithmetics PA~\cite{cohen2015middle}.

\subsection{Infinitary Proof Systems}
\label{sec:ProofSystems:Infinitary}

\begin{definition}
  The infinitary proof system {\ortc} for {\rtclogic} is defined like {\rtc}, but replacing \cref{eq:RTC-min} by:
  \begin{equation}
    \label[rule]{eq:RTC-cyc}
    \begin{gathered}
    \infer
      {\sequent
        {\Gamma, (\rtcformula{x}{y}{\varphi})(s, t)}
        {\Delta}}
      { \sequent{\Gamma, s=t}{\Delta}
          &
        \sequent
          {\Gamma, (\rtcformula{x}{y}{\varphi})(s, z), {\varphi \subst{\scalebox{0.8}{\replace{z}{x}, \replace{t}{y}}}}}
          {\Delta}}
    \end{gathered}
  \end{equation}
  where $z$ is fresh, i.e.~does not occur free in $\Gamma$, $\Delta$, or $(\rtcformula{x}{y}{\varphi})(s, t)$. 
The formula $(\rtcformula{x}{y}{\varphi})(s, z)$ in the right-hand premise  is called the \emph{immediate ancestor} (cf.~\cite[{\textsection}1.2.3]{buss1998handbook}) of the principal formula, $(\rtcformula{x}{y}{\varphi})(s, t)$, in the conclusion.
\end{definition}

There is an asymmetry between \cref{eq:transRTC}, in which the intermediary is an arbitrary term $r$, and \cref{eq:RTC-cyc}, where we use a variable $z$. This is necessary to obtain the soundness of the cyclic proof system. It is used to show that when there is a counter-model for the conclusion of a rule, then there is also a counter-model for one of its premises that is, in a sense that we make precise below, `smaller'. In the case that $s \neq t$, using a fresh $z$ allows us to pick from \emph{all} possible counter-models of the conclusion, from which we may then construct the required counter-model for the right-hand premise. If we allowed an arbitrary term $r$ instead, this might restrict the counter-models we can choose from, only leaving ones `larger' than the one we had for the conclusion. See \cref{lem:DescendingCountermodels} below for more details.

Proofs in this system are possibly infinite derivation trees. However, not all infinite derivations are proofs: only those that admit an infinite descent argument. Thus we use the terminology `pre-proof' for derivations.

\begin{definition}[Pre-proofs]
  An {\ortc} \emph{pre-proof} is a possibly infinite (i.e. non-well-founded) derivation tree formed using the inference rules.
  A \emph{path} in a pre-proof is a possibly infinite sequence of sequents $s_0, s_1, \ldots (, s_n)$ such that $s_0$ is the root sequent of the proof, and $s_{i+1}$ is a premise of $s_i$ for each $i < n$.
\end{definition}

The following definitions tell us how to track {\rtcOperator} formulas through a pre-proof, and allow us to formalize inductive arguments via infinite descent.

\begin{definition}[Trace Pairs]
Let $\tau$ and $\tau'$ be {\rtcOperator} formulas occurring in the left-hand side of the conclusion $s$ and a premise $s'$, respectively, of (an instance of) an inference rule. $(\tau, \tau')$ is said to be a \emph{trace pair} for $(s, s')$ if the rule is:
  \begin{itemize}
    \item
    the (Subst) rule, and $\tau = \tau' \theta$ where $\theta$ is the substitution associated with the rule instance;
    \item
    Rule \eqref{eq:RTC-cyc}, and either:
    \begin{enumerate}[label={\alph*)}]
      \item
      $\tau$ is the principal formula of the rule instance and $\tau'$ is the immediate ancestor of $\tau$, in which case we say that the trace pair is \emph{progressing};
      \item
      otherwise, $\tau = \tau'$.
    \end{enumerate}
    \item
    any other rule, and $\tau = \tau'$.
  \end{itemize}
\end{definition}

\begin{definition}[Traces]
  A \emph{trace} is a (possibly infinite) sequence of {\rtcOperator} formulas. We say that a trace $\tau_{1}, \tau_{2}, \ldots (, \tau_{n})$ \emph{follows} a path $s_{1}, s_{2}, \ldots (, s_{m})$ in a pre-proof $\mathcal{P}$ if, for some $k \geq 0$, each consecutive pair of formulas $(\tau_{i}, \tau_{i+1})$ is a trace pair for $(s_{i+k}, s_{i+k+1})$.
  If $(\tau_{i}, \tau_{i+1})$ is a progressing pair then we say that the trace \emph{progresses} at $i$, and we say that the trace is \emph{infinitely progressing} if it progresses at infinitely many points.
\end{definition}

Proofs, then, are pre-proofs which satisfy a global trace condition.

\begin{definition}[Infinite Proofs]
  A {\ortc} \emph{proof} is a pre-proof in which every infinite path is followed by some infinitely progressing trace.
\end{definition}

Clearly, we cannot reason effectively about such infinite proofs in general. In order to do so we need to restrict our attention to those proof trees which are finitely representable. These are the \emph{regular} infinite proof trees, which contain only finitely many \emph{distinct} subtrees. They can be specified as systems of recursive equations or, alternatively, as cyclic \emph{graphs} \cite{Courcelle83}. Note that a given regular infinite proof may have many different graph representations. One possible way of formalizing such proof graphs is as standard proof trees containing open nodes (called buds), to each of which is assigned a syntactically equal internal node of the proof (called a companion). Due to space limitation, we elide a formal definition of cyclic proof graphs (see, e.g., Sect.~7 in \cite{brotherston2010sequent}) and rely on the reader's basic intuitions.

\begin{definition}[Cyclic Proofs]
  The cyclic proof system {\cortc} for {\rtclogic} is the subsystem of {\ortc} comprising of all and only the finite and \emph{regular} infinite proofs (i.e.~those proofs that can be represented as finite, possibly cyclic, graphs). 
\end{definition}

Note that it is decidable whether a cyclic pre-proof satisfies the global trace condition, using a construction involving an inclusion between B\"{u}chi automata (see, e.g., \cite{Brotherston07,Simpson2017}). However since this requires complementing B\"{u}chi automata (a \textsf{PSPACE} procedure), our system cannot be considered a proof system in the Cook-Reckhow sense \cite{CookReckhow79}. Notwithstanding, checking the trace condition for cyclic proofs found in practice is not prohibitive \cite{RoweBrotherston17,TellezBrotherston17}.

\subsection{Soundness and Completeness}
\label{sec:comp}

The rich expressiveness of {\TC} logic entails that the effective system {\rtc} which is sound w.r.t. the standard semantics, cannot be complete (much like the case for {\lkid}). It is however both sound and complete w.r.t. Henkin semantics. 

\begin{theorem}[Soundness and Completeness of {\rtc} \cite{Cohen2017Henkin}] 
  \label{rtc_comp}
  {\rtc} is sound for standard semantics, and also sound and complete for Henkin semantics. % \hfill \\
  \lcnote{the proof in \cite{Cohen2017Henkin} is for closed formulas. can probably derive for open. I think working with closed sequent is the standard. (is LKID proof in \cite{brotherston2010sequent} for open sequents?yes!)}
\end{theorem}

Note that the system {\rtc} as presented here does not admit cut elimination. The culprit is the induction rule \eqref{eq:RTC-min}, which does not permute with cut. We may obtain admissibility of cut by using the following alternative formulation of the induction rule:
    \begin{equation*}
    \infer
    {\sequent
      {\Gamma,  (\rtcformula{x}{y}{\varphi})(s, t)} {\Delta}}
    {\sequent
      {\Gamma}{{\psi \subst{\replace{s}{x}}}}
      \quad
       \sequent
      {\Gamma, \psi(x), \varphi(x,y)}{{\psi \subst{\replace{y}{x}}}}
     \quad 
     \sequent
      {\Gamma,{\psi \subst{\replace{t}{x}}}}{\Delta}}
    \end{equation*}
where $x \not\in \fn{fv}(\Gamma, \Delta)$, $y \not\in \fn{fv}(\Gamma, \Delta, \psi)$.
Like the induction rule for {\lkid}, this formulation incorporates a cut with the induction formula $\psi$.
For the system with this rule, a simple adaptation of the completeness proof in \cite{Cohen2017Henkin}, in the spirit of the corresponding proof for {\lkid} in \cite{brotherston2010sequent}, suffices to obtain cut-free completeness. However, the tradeoff is that the resulting cut-free system no longer has the sub-formula property. In contrast, cut-free proofs in {\rtc} do satisfy the sub-formula property, for a generalized notion of a subformula that incorporates substitution instances (as in $\LK_=$).

We remark that the soundness proof of {\lkid} is rather complex since it must handle different types of mutual dependencies between the inductive predicates. 
%It is provided in the Appendix of\cite{brotherston2010sequent}. 
For {\rtc} the proof is much simpler due to the uniformity of the rules for the {\rtcOperator} operator.

The infinitary system {\ortc}, in contrast to the finitary system {\rtc}, is both sound and complete w.r.t. the standard semantics. To prove soundness, we make use of the following notion of \emph{measure} for {\rtcOperator} formulas.

\begin{definition}[Degree of {\rtcOperator} Formulas]
  For $\phi \equiv (\rtcformula{x}{y}{\varphi})(s,t)$, we define $\degreeOf{\phi}{M}{v} = 0$ if $v(s) = v(t)$, and $\degreeOf{\phi}{M}{v} = n$ if $v(s) \neq v(t)$ and $a_0, \ldots, a_n$ is a minimal-length sequence of elements in the semantic domain $D$ such that $v(s) = a_{0}$, $v(t) = a_{n}$, and $M, v[x := a_{i}, y := a_{i+1}] \models \varphi$ for $0 \leq i < n$. We call $\degreeOf{\phi}{M}{v}$ the \emph{degree} of $\phi$ with respect to the model $M$ and valuation $v$.
\end{definition}

Soundness then follows from the following fundamental lemma.

\begin{restatable}[Descending Counter-models]{lemma}{descendingcountermodels}
  \label[lemma]{lem:DescendingCountermodels}
  If there exists a standard model $M$ and valuation $v$ that invalidates the conclusion $s$ of (an instance of) an inference rule, then
  \begin{enumerate*}[label={\arabic*)}]
    \item
    there exists a standard model $M'$ and valuation $v'$ that invalidates some premise $s'$ of the rule; and
    \item
    if $(\tau, \tau')$ is a trace pair for $(s, s')$ then $\degreeOf{\tau'}{M'}{v'} \leq \degreeOf{\tau}{M}{v}$. Moreover, if $(\tau, \tau')$ is a progressing trace pair then $\degreeOf{\tau'}{M'}{v'} < \degreeOf{\tau}{M}{v}$.
  \end{enumerate*}
\end{restatable}
\begin{proof}
  The cases for the standard $\LK_{=}$ and substitution rules are straightforward adaptations of those found in e.g.~\cite{brotherston2010sequent}. 
  \begin{itemize}[wide,itemsep={0.75em}]
    \item
    The case for Rule \eqref{eq:refRTC} follows trivially since it follows immediately from \cref{def:Semantics:Standard} that $M, v \models (\rtcformula{x}{y}{\varphi})(s, s)$ for all $M$ and $v$. 
    \item
    For Rule \eqref{eq:transRTC}, since $M, v \not\models (\rtcformula{x}{y}{\varphi})(s, t)$ it follows that either $M, v \not\models (\rtcformula{x}{y}{\varphi})(s, r)$ or $M, v \not\models \varphi \subst{\scalebox{0.865}{\replace{r}{x}, \replace{t}{y}}}$. To see this, suppose for contradiction that both $M, v \models (\rtcformula{x}{y}{\varphi})(s, r)$ or $M, v \models \varphi \subst{\scalebox{0.865}{\replace{r}{x}, \replace{t}{y}}}$; but then it would follow by \cref{def:Semantics:Standard} that $M, v \models (\rtcformula{x}{y}{\varphi})(s, t)$. We thus take $M' = M$ and $v' = v$, and either the left- or right-hand premise according to whether $M, v \not\models (\rtcformula{x}{y}{\varphi})(s, r)$ or $M, v \not\models \varphi \subst{\scalebox{0.865}{\replace{r}{x}, \replace{t}{y}}}$. 
    \item
    For Rule \eqref{eq:RTC-cyc}, since $M, v \models (\rtcformula{x}{y}{\varphi})(s, t)$ there are two cases to consider:
    \begin{enumerate}[nosep,label={(\roman*)}]
      \item
      If $v(s) = v(t)$ then we take the left-hand premise with model $M' = M$ and valuation $v' = v$, and so the degree of any {\rtcOperator} formula in $\Gamma$ with respect to $M'$ and $v'$ remains the same.
      \item      
      If on the other hand there are $a_0, \ldots, a_n \in D$ ($n > 0$) such that $v(s) = a_0$ and $v(t) = a_n$ with $M, v[x := a_{i}, y := a_{i+1}] \models \varphi$ for $0 \leq i < n$, we then take the right-hand premise, the model $M' = M$ and valuation $v' = v[z := a_{n-1}]$. Note that, without loss of generality, we may assume a sequence $a_0, \ldots, a_n$ of minimal length, and thus surmise $\degreeOf{(\rtcformula{x}{y}{\varphi})(s, t)}{M}{v} = n$. Since $z$ is fresh, it follows that $M', v' \models \varphi \subst{\scalebox{0.865}{\replace{z}{x}, \replace{t}{y}}}$ and $M', v'[x := a_{i}, y := a_{i+1}] \models \varphi$ for $0 \leq i < n - 1$. If $n = 1$ then $v'(s) = v'(z) = a_0$ and so $M, v' \models (\rtcformula{x}{y}{\varphi})(s, z)$; otherwise this is witnessed by the sequence $a_{0}, \ldots, a_{n-1}$. Thus we also have that $\degreeOf{(\rtcformula{x}{y}{\varphi})(s, z)}{M'}{v'} = n - 1$. To conclude, note it also follows from $z$ fresh that $M', v' \models \psi$ for all $\psi \in \Gamma$ and $M', v' \not\models \phi$ for all $\phi \in \Delta$; and furthermore that the degree of any {\rtcOperator} formula in $\Gamma$ remains unchanged with respect to $M'$ and $v'$.
      \qed
    \end{enumerate}
  \end{itemize}
\end{proof}

As is standard for infinite descent inference systems \cite{Brotherston07,BrotherstonBC08,brotherston2010sequent,DasP17,RoweBrotherston17,TellezBrotherston17}, the above result entails the local soundness of the inference rules (in our case, for standard first-order models). The presence of infinitely progressing traces for each infinite path in a {\ortc} proof ensures soundness via a standard infinite descent-style construction.

\begin{restatable}[Soundness of {\ortc}]{theorem}{infinitarysoundness}
  \label{sound}
  If there is a {\ortc} proof of $\sequent{\Gamma}{\Delta}$, then $\sequent{\Gamma}{\Delta}$ is valid (w.r.t.~the standard semantics)
\end{restatable}
\begin{proof}
  Suppose, for contradiction, that $\sequent{\Gamma}{\Delta}$ is not valid. Then by \cref{lem:DescendingCountermodels} there exists an infinite path $\seq{s_{i}}{i > 0}$ in the proof and an infinite sequence of model-valuation pairs $\seq{\pair{M_{i}}{v_{i}}}{i > 0}$ such that $\pair{M_i}{v_i}$ invalidates $s_i$ for each $i > 0$. Since the proof is a valid {\ortc} proof, this infinite path is followed by an infinitely progressing trace $\seq{\tau_{i}}{i > 0}$ for which we can take the degree of each formula with respect to its corresponding counter-model to obtain an infinite sequence of natural numbers $\seq{\degreeOf{\tau_{i}}{M_{k+i}}{v_{k+i}}}{i > 0}$ (for some $k \geq 0$). By \cref{lem:DescendingCountermodels} this sequence is decreasing and, moreover, since the trace is infinitely progressing the sequence strictly decreases infinitely often. From the fact that the natural numbers are a well-founded set we derive a contradiction, and thus conclude that $\sequent{\Gamma}{\Delta}$ is indeed valid.
  \qed
\end{proof}

The soundness of the cyclic system is an immediate corollary, since each {\cortc} proof is also a {\ortc} proof.

\begin{corollary}[Soundness of {\cortc}]
  \label[corollary]{cor:cortc-implies-ortc}
  If there is a {\cortc} proof of $\sequent{\Gamma}{\Delta}$, then $\sequent{\Gamma}{\Delta}$ is valid (w.r.t.~the standard semantics) 
\end{corollary}

Following a standard technique (as used in e.g.~\cite{brotherston2010sequent}), we can show cut-free completeness of {\ortc} with respect to the standard semantics.

\begin{definition}[Schedule]
  A \emph{schedule element} $E$ is defined as any of the following:
  \begin{itemize}[nosep]
  \item
  a formula of the form $\neg \varphi, \varphi \wedge \psi, \varphi \vee \psi, \varphi \rightarrow \psi$;
  \item
  a pair of the form $\pair{\forall x \, \varphi}{t}$ or $\pair{\exists x \, \varphi}{t}$ where $\forall x \, \varphi$ and $\exists x \, \varphi$ are formulas and $t$ is a term;
  \item
  a tuple of the form $\langle (\rtcformula{x}{y}{\varphi})(s, t), r, z, \Gamma, \Delta \rangle$ where $(\rtcformula{x}{y}{\varphi})(s, t)$ is a formula, $r$ is a term, $\Gamma$ and $\Delta$ are finite sequences of formulas, and $z$ is a variable not occurring free in $\Gamma$, $\Delta$, or $(\rtcformula{x}{y}{\varphi})(s, t)$; or
  \item
  a tuple of the form $\langle s = t, x, \varphi, n, \Gamma, \Delta \rangle$ where $s$ and $t$ are terms, $x$ is a variable, $\varphi$ is a formula, $n \in \{1, 2\}$, and $\Gamma$ and $\Delta$ are finite sequences of formulas.
  \end{itemize}
A \emph{schedule} is a recursive enumeration of schedule elements in which every schedule element appears infinitely often (these exist since our language is countable).
\end{definition}

Each schedule corresponds to an exhaustive search strategy for a cut-free proof for each sequent $\Gamma \Rightarrow \Delta$, via the following notion of a `search tree'.

\begin{definition}[Search Tree]
  Given a schedule $\seq{E_{i}}{i > 0}$, for each sequent $\Gamma \Rightarrow \Delta$ we inductively define an infinite sequence of (possibly open) derivation trees, $\seq{T_{i}}{i > 0}$, such that $T_{1}$ consists of the single open node $\Gamma \Rightarrow \Delta$, and each $T_{i+1}$ is obtained by replacing all suitable open nodes in $T_{i}$ with applications of first axioms and then the left and right inference rules for the formula in the $i$\textsuperscript{th} schedule element.
  We show the cases for building $T_{i+1}$ for when $E_i$ corresponds to an {\rtcOperator} formula and an equality formula. The cases for when $E_i$ corresponds to a standard compound first-order formula are similar.
  \begin{itemize}[wide]
    \item
    When $E_{i}$ is of the form $\langle (\rtcformula{x}{y}{\varphi})(s, t), r, z, \Gamma, \Delta \rangle$, then $T_{i+1}$ is obtained by:
    \begin{enumerate}[wide,topsep={\smallskipamount},itemsep={\smallskipamount}]
      \item
      first closing as such any open node that is an instance of an axiom (after left and right weakening, if necessary);
      \item
      next, replacing every open node $\sequent{\Gamma', (\rtcformula{x}{y}{\varphi})(s, t)}{\Delta'}$ of the resulting tree for which $\Gamma' \subseteq \Gamma$ and $\Delta' \subseteq \Delta$ with the derivation:
      \par\noindent
      \begin{minipage}{\textwidth}
      \smaller[2]
      \begin{equation*}
        \infer[{\eqref{eq:RTC-cyc}}]
          {\sequent{\Gamma', (\rtcformula{x}{y}{\varphi})(s, t)}{\Delta'}}
          {\sequent{\Gamma', (\rtcformula{x}{y}{\varphi})(s, t), s = t}{\Delta'}
           \qquad
            \sequent
              {\Gamma', (\rtcformula{x}{y}{\varphi})(s, t), (\rtcformula{x}{y}{\varphi})(s, z), {\varphi \subst{\replace{z}{x}, \replace{t}{y}}}}
              {\Delta'}
           }
      \end{equation*}
      \end{minipage}
      \smallskip
      \item
      finally, replacing every open node $\sequent{\Gamma'}{\Delta', (\rtcformula{x}{y}{\varphi})(s, t)}$ of the resulting tree with the derivation:
      \medskip
      \par\noindent
      \begin{minipage}{\textwidth}
      \smaller[2]
      \begin{equation*}
        \infer[{\eqref{eq:transRTC}}]
          {\sequent{\Gamma'}{\Delta', (\rtcformula{x}{y}{\varphi})(s, t)}}
          {\sequent{\Gamma'}{\Delta', (\rtcformula{x}{y}{\varphi})(s, t), (\rtcformula{x}{y}{\varphi})(s, r)}
            \quad
           \sequent{\Gamma'}{\Delta', (\rtcformula{x}{y}{\varphi})(s, t), {\varphi \subst{\replace{r}{x}, \replace{t}{y}}}}}
      \end{equation*}
      \end{minipage}
      \smallskip
    \end{enumerate}
    \item
    When $E_i$ is of the form $\langle s = t, x, \varphi, n, \Gamma, \Delta \rangle$, then $T_{i+1}$ is then obtained by first closing as such any open node that is an instance of an axiom (after left and right weakening, if necessary); and next, replacing every open node $\sequent{\Gamma', s = t}{\Delta', \psi}$ in the resulting tree where $\Gamma' \subseteq \Gamma$ and $\Delta' \subseteq \Delta$, and $\psi$ is $\varphi \subst{\replace{s}{x}}$ (resp.~$\varphi \subst{\replace{t}{x}}$) if $n = 1$ (resp.~$n = 2$), with the appropriate one of the following derivations:
    \medskip
    \par\noindent
    \begin{minipage}{\textwidth}
    \smaller
    \begin{gather*}
      \infer[({=}\text{L}_1)]
        {\sequent{\Gamma', s = t}{\Delta', {\varphi \subst{\replace{t}{x}}}}}
        {\sequent{\Gamma', s = t}{\Delta', {\varphi \subst{\replace{t}{x}}}, {\varphi \subst{\replace{s}{x}}}}}
      \hspace{4em}
      \infer[({=}\text{L}_2)]
        {\sequent{\Gamma', s = t}{\Delta', {\varphi \subst{\replace{s}{x}}}}}
        {\sequent{\Gamma', s = t}{\Delta', {\varphi \subst{\replace{s}{x}}}, {\varphi \subst{\replace{t}{x}}}}}
    \end{gather*}
    \end{minipage}
    \smallskip
  \end{itemize}
  The limit of the sequence $\seq{T_{i}}{i > 0}$ is a possibly infinite (and possibly open) derivation tree called the \emph{search tree} for $\sequent{\Gamma}{\Delta}$ with respect to the schedule $\seq{E_{i}}{i > 0}$, and denoted by $T_\omega$.
\end{definition}

Search trees are, by construction, recursive and cut-free. We construct special `sequents' out of search trees, called \emph{limit sequents}, as follows.

\begin{definition}[Limit Sequents]
\label[definition]{def:LimitSequent}
  When a search tree $T_\omega$ is not an {\ortc} proof, either:
  \begin{enumerate*}[label=({\arabic*})]
    \item \label{def:LimitSequent:OpenNode}
    it is not even a pre-proof, i.e. it contains an open node; or 
    \item \label{def:LimitSequent:UntraceableBranch}
    it is a pre-proof but contains an infinite branch that fails to satisfy the global trace condition.
  \end{enumerate*}
  In case \ref{def:LimitSequent:OpenNode} it contains an open node to which, necessarily, no schedule element applies (e.g.~a sequent containing only atomic formulas), for which we write $\sequent{\Gamma_\omega}{\Delta_\omega}$. 
  In case \ref{def:LimitSequent:UntraceableBranch} the global trace condition fails, so there exists an infinite path $\seq{\sequent{\Gamma_i}{\Delta_i}}{i > 0}$ in $T_\omega$ which is followed by no infinitely progressing traces; we call this path the \emph{untraceable branch} of $T_\omega$.
  We then define $\Gamma_{\omega} = \bigcup_{i > 0}{\Gamma_{i}}$ and $\Delta_{\omega} = \bigcup_{i > 0}{\Delta_{i}}$, and call $\sequent{\Gamma_{\omega}}{\Delta_{\omega}}$ the \emph{limit sequent}.\footnote{To be rigorous, we may pick e.g.~the left-most open node or untraceable branch.}
\end{definition}

Note that use of the word `sequent' here is an abuse of nomenclature, since limit sequents may be infinite and thus technically not sequents. However when we say that such a limit sequent is provable, we mean that it has a finite subsequent that is provable.

\begin{restatable}{lemma}{LimitSequentsNotProvable}
\label{not-prov}
  Limit sequents $\sequent{\Gamma_\omega}{\Delta_\omega}$ are not cut-free provable.
\end{restatable}
\begin{proof}
  Straightforward adaptation of the proof of \cite[Lemma~6.3]{brotherston2010sequent}.
\end{proof}

As standard, we use a limit sequent to induce a counter-interpretation, consisting of a Herbrand model quotiented by the equalities found in the limit sequent.% which are counter-models for invalid sequents.

\begin{definition}[Quotient Relation]
  For a limit sequent $\sequent{\Gamma_{\omega}}{\Delta_{\omega}}$, the relation $\sim$ is defined as the smallest congruence relation on terms such that $s \sim t$ whenever $s = t \in \Gamma_{\omega}$. We write $[t]$ for the $\sim$-equivalence class of $t$, i.e.~$[t] = \{ u \,\mid\, t \sim u \}$.
\end{definition}

The following property holds of the quotient relation.

\begin{lemma}
  \label{lem:quotient}
  If $t \sim u$, then $\sequent{\Gamma_{\omega}}{F \subst{\replace{t}{x}}}$ is cut-free provable in {\ortc} if and only if $\sequent{\Gamma_{\omega}}{F \subst{\replace{u}{x}}}$.
\end{lemma}
\begin{proof}
  By induction on the conditions defining $\sim$. We use $\equiv$ to denote syntactic equality on terms, in order to distinguish from formulas $s = t$ asserting equality between (interpretations of) terms.
  \begin{description}
    \item[($t \sim t$):]
    Immediate, since then $t \equiv u$.
    \item[($t = u \in \Gamma_{\omega}$):]
    Assume $\sequent{\Gamma_{\omega}}{F \subst{\replace{t}{x}}}$ is cut-free provable, then we can apply the (${=}\text{L}_1$) rule to derive (without cut) $\sequent{\Gamma_{\omega}, t = u}{F \subst{\replace{u}{x}}}$; however notice that $\Gamma_{\omega}, t = u$ is simply $\Gamma_{\omega}$ since $t = u \in \Gamma_{\omega}$ already. The converse direction is symmetric, using rule (${=}\text{L}_2$).
    \item[($t \sim u \Rightarrow u \sim t$):]
    Immediate, by induction.
    \item[($t \sim u \wedge u \sim v \Rightarrow t \sim v$):]
    Straightforward, by induction.
    \item[($t_1 \sim u_1 \wedge \ldots \wedge t_n \sim u_n \Rightarrow f(t_1, \ldots, t_n) \sim f(u_1, \ldots, u_n)$):]
    Consider the formula $F$; clearly there exist formulas $G_1, \ldots, G_n$ and some variable $y$ such that ${G_i \subst{\scalebox{0.865}{\replace{t}{y}}}} \equiv {F \subst{\scalebox{0.865}{\replace{f(u_1, \ldots, u_{i-1}, t_i, \ldots, t_n)}{x}}}}$ for each $i \leq n$.
    By induction, each sequent $\sequent{\Gamma_{\omega}}{G_i \subst{\scalebox{0.865}{\replace{t_i}{y}}}}$ is cut-free provable if and only if so too is $\sequent{\Gamma_{\omega}}{G_i \subst{\scalebox{0.865}{\replace{u_i}{y}}}}$.
    The result then follows since ${F \subst{\scalebox{0.865}{\replace{f(t_1, \ldots, t_n)}{x}}}} \equiv {G_1 \subst{\scalebox{0.865}{\replace{t_1}{y}}}}$ and ${F \subst{\scalebox{0.865}{\replace{f(u_1, \ldots, u_n)}{x}}}} \equiv {G_n \subst{\scalebox{0.865}{\replace{u_n}{y}}}}$, and also ${G_i \subst{\scalebox{0.865}{\replace{u_i}{y}}}} \equiv {G_{i+1} \subst{\scalebox{0.865}{\replace{t_{i+1}}{y}}}}$ for each $i < n$.
  \end{description}
\end{proof}

We define the counter-interpretation as follows.

\begin{definition}[Counter-interpretations]
  \label[definition]{def:Counter-interpretations}
  Assume a search tree $T_{\omega}$ which is not a {\ortc} proof with limit sequent $\sequent{\Gamma_{\omega}}{\Delta_{\omega}}$.
  Define a structure $M_{\omega} = \langle D, I \rangle$ as follows:
  \begin{itemize}[nosep]
    \item 
    $D = \{ [t] \,\mid\, \text{t is a term} \}$ (i.e.~the set of terms quotiented by the relation $\sim$).%; that is, the set of $\sim$-equivalence classes of terms.
    \item
    For every $k$-ary function symbol $f$: 
    $I(f)([t_1], \ldots, [t_k]) = [f (t_1, \ldots, t_k)]$
    \item 
    For every $k$-ary relation symbol $q$:
    $I(q) = \{ ([t_1], \ldots, [t_k]) \,\mid\, q(t_1, \ldots, t_k) \in \Gamma_{\omega} \}$
  \end{itemize}
  We also define a valuation $\rho_{\omega}$ for $M_{\omega}$ by $\rho_{\omega}(x) = [x]$ for all variables $x$. 
%  \cutnote{cut}We call $(M_{\omega}, \rho_{\omega})$ the \emph{counter-interpretation} for $\sequent{\Gamma_{\omega}}{\Delta_{\omega}}$.
\end{definition}

The counter-interpretation $\pair{M_{\omega}}{\rho_{\omega}}$ has the following property, meaning that $M_{\omega}$ is a counter-model for the corresponding sequent $\sequent{\Gamma}{\Delta}$ if its search tree $T_{\omega}$ is not a proof.

\begin{restatable}{lemma}{LimitSequentProperties}
  \label[lemma]{lem:LimitSequents}
  If $\psi \in \Gamma_{\omega}$ then $M_{\omega}, \rho_{\omega} \models \psi$; and if $\psi \in \Delta_{\omega}$ then $M_{\omega}, \rho_{\omega} \not\models \psi$.
\end{restatable}
\begin{proof}
  By well-founded induction using the lexicographic ordering of the number of binders (i.e.~$\exists$, $\forall$, and {\rtcOperator}) in $\psi$ and the structure of $\psi$. 
  Notice that, by definition, $\rho_{\omega}(t) = [t]$ for all terms $t$.
  
  \medskip
  For $\psi$ atomic (i.e.~of the form $q(t_{1}, \ldots, t_{k})$), if $\psi \in \Gamma_{\omega}$ then it follows immediately by \Cref{def:Counter-interpretations} that $M_{\omega}, \rho_{\omega} \models q(t_{1}, \ldots, t_{k})$. 
  If, on the other hand, $\psi \in \Delta_{\omega}$ then assume for contradiction that indeed $M_{\omega}, \rho_{\omega} \models q(t_{1}, \ldots, t_{k})$. 
  It then follows from \cref{def:Counter-interpretations} that there are terms $u_1, \ldots, u_k$ such that $q(u_1, \ldots, u_k) \in \Gamma_{\omega}$ and $u_i \sim t_i$ for each $i \leq k$.
  Notice that then we can prove $\sequent{\Gamma_{\omega}}{q(u_1, \ldots, u_k)}$ axiomatically, and so it follows by ($k$ applications of) \cref{lem:quotient} that $\sequent{\Gamma_{\omega}}{q(t_1, \ldots, t_k)}$ is cut-free provable.
  However, since $q(t_1, \ldots, t_k) \in \Delta_{\omega}$, this would mean that the limit sequent $\sequent{\Gamma_{\omega}}{\Delta_{\omega}}$ is cut-free provable, which contradicts \cref{not-prov}.
  Thus we conclude that in fact $M_{\omega}, \rho_{\omega} \not\models q(t_{1}, \ldots, t_{k})$.
  
  \medskip
  For $\psi$ an equality formula $s = t$, if $\psi \in \Gamma_{\omega}$ then we have immediately by \cref{def:Counter-interpretations} that $\rho_{\omega}(s) = \rho_{\omega}(t)$ and thus that $M_{\omega}, \rho_{\omega} \models s = t$ by \cref{def:Semantics:Standard}.
  If, on the other hand, $s = t \in \Delta_{\omega}$, suppose for contradiction that indeed $M_{\omega}, \rho_{\omega} \models s = t$.
  It then follows from \cref{def:Counter-interpretations} that $s \sim t$.
  Since we may derive $\sequent{\Gamma_{\omega}}{s = s}$ axiomatically, it thus follows from \cref{lem:quotient} that there is a cut-free proof of $\sequent{\Gamma_{\omega}}{s = t}$.
  However, since $s = t \in \Delta_{\omega}$ this would mean that the limit sequent $\sequent{\Gamma_{\omega}}{\Delta_{\omega}}$ is cut-free provable, which contradicts \cref{not-prov}.
  We thus conclude that in fact $M_{\omega}, \rho_{\omega} \not\models s = t$.

  \medskip
  The cases where $\psi$ is a standard compound first-order formula follow easily by induction.

  \medskip
  In case $\psi = (\rtcformula{x}{y}{\varphi})(s, t)$, we reason as follows.
  \begin{itemize}[wide,itemsep={\medskipamount}]
    \item
    For the first part of the lemma assume $(\rtcformula{x}{y}{\varphi})(s, t)\in \Gamma_\omega$. 
    Then, by the construction of $T_{\omega}$, there is at least one occurrence of rule \eqref{eq:RTC-cyc} with active formula $\psi$ in the untraceable branch; thus there a two cases:
    \begin{enumerate}[label={\roman*)},topsep={\smallskipamount}] 
      \item
      The branch follows the left-hand premise, so there is $s = t \in \Gamma_{\omega}$.
      Therefore, by \cref{def:Counter-interpretations}, $\rho_{\omega}(s) = \rho_{\omega}(t)$ and so it follows immediately from \cref{def:Semantics:Standard} that $M_{\omega}, \rho_{\omega} \models (\rtcformula{x}{y}{\varphi})(s, t)$.
%      Then  and $M_{\omega}, \rho_{\omega} \models \psi$ follows as above; or
      \item
      The branch follows the right-hand premise and, since there is no infinitely progressing trace along the untraceable branch, there must be a finite number of distinct variables $z_1, \ldots, z_n$ ($n > 0$) such that $\varphi \subst{\scalebox{0.865}{\replace{z_{i}}{x}, \replace{z_{i+1}}{y}}} \in \Gamma_{\omega}$, for each $i < n$, and $\varphi \subst{\scalebox{0.865}{\replace{z_{n}}{x}, \replace{t}{y}}} \in \Gamma_{\omega}$.
      Then, by the I.H., $M_{\omega}, \rho_{\omega} \models \varphi \subst{\scalebox{0.865}{\replace{z_{i}}{x}, \replace{z_{i+1}}{y}}}$ for each $i < n$, and $M_{\omega}, \rho_{\omega} \models \varphi \subst{\scalebox{0.865}{\replace{z_{n}}{x}, \replace{t}{y}}}$.
      Thus, $M_{\omega}, \rho_{\omega} [x := [z_{i}], y := [z_{i+1}]] \models \varphi$ for each $i < n$, and $M_{\omega}, \rho_{\omega} [x := [z_{n}], y := [t]] \models \varphi$.
      Moreover, the untraceable branch also follows the left-hand premise of rule \eqref{eq:RTC-cyc} with active formula $(\rtcformula{x}{y}{\varphi})(s, z_{1})$.
      Thus $s = z_{1} \in \Gamma_{\omega}$, and so $\rho_{\omega}(s) = \rho_{\omega}(z_{1}) = [z_{1}]$.
      We then have from \Cref{def:Semantics:Standard} that $M_{\omega}, \rho_{\omega} \models \psi$.
    \end{enumerate}
    \item
    For the second part of the lemma we first prove, by an inner induction on $n$, the following auxiliary result for all terms $s$ and $t$ and elements $a_0, \ldots, a_n \in D$ ($n > 0$):
    \medskip
    \begin{adjustwidth}{1em}{1em}
      if $(\rtcformula{x}{y}{\varphi})(s, t) \in \Delta_{\omega}$, with $\rho_{\omega}(s) = a_0$ and $\rho_{\omega}(t) = a_n$, then there exists some $i < n$ such that $M_{\omega}, \rho_{\omega} [x := a_{i}, y := a_{i+1}] \not\models \varphi$.
    \end{adjustwidth}
    \smallskip
    \begin{description}
      \item[($n = 1$):]
      Since $(\rtcformula{x}{y}{\varphi})(s, t) \in \Delta_{\omega}$, we have $\varphi \subst{\scalebox{0.865}{\replace{s}{x}, \replace{t}{y}}} \in \Delta_{\omega}$ by construction as the untraceable branch must traverse an instance of rule \eqref{eq:transRTC} with $r \equiv s$ and moreover must traverse the right-hand premise (otherwise, we would have $(\rtcformula{x}{y}{\varphi})(s, s) \in \Delta_{\omega}$ resulting in the branch being closed by an instance of rule \eqref{eq:refRTC}). Thus by the outer induction it follows that $M_{\omega}, \rho_{\omega} \not\models \varphi \subst{\scalebox{0.865}{\replace{s}{x}, \replace{t}{y}}}$ and thence that $M_{\omega}, \rho_{\omega} [x := \rho_{\omega}(s), y := \rho_{\omega}(t)] \not\models \varphi$ as required.
      \item[($n = k + 1$, $k > 0$):]
      Then there exists some term $r$ such that $a_{k} = [r] = \rho_{\omega}(r)$. If we have $(\rtcformula{x}{y}{\varphi})(s, t) \in \Delta_{\omega}$, by construction of the search tree $T_{\omega}$ we then also have that either $(\rtcformula{x}{y}{\varphi})(s, r) \in \Delta_{\omega}$ or $\varphi \subst{\scalebox{0.865}{\replace{r}{x}, \replace{t}{y}}} \in \Delta_{\omega}$, as the untraceable branch must traverse an instance of rule \eqref{eq:transRTC} for the term $r$. In the case of the former, the required result holds by the inner induction. In the case of the latter, we have $M_{\omega}, \rho_{\omega} \not\models \varphi \subst{\scalebox{0.865}{\replace{r}{x}, \replace{t}{y}}}$ by the outer induction and thence that $M_{\omega}, \rho_{\omega} [x := \rho_{\omega}(r), y := \rho_{\omega}(t)] \not\models \varphi$; i.e.~$M_{\omega}, \rho_{\omega} [x := a_{k}, y := a_{k+1}] \not\models \varphi$ as required.
    \end{description}
    We now show that the primary result holds. Assume $(\rtcformula{x}{y}{\varphi})(s, t) \in \Delta_{\omega}$ and suppose for contradiction that $M_{\omega}, \rho_{\omega} \models (\rtcformula{x}{y}{\varphi})(s, t)$ holds. Thus, by \cref{def:Semantics:Standard}, there are two cases to consider.
    \begin{itemize}[topsep={\medskipamount}]
      \item
      If $\rho_{\omega}(s) = \rho_{\omega}(t)$ then $s \sim t$.
      Thus since we may derive $\sequent{\Gamma_{\omega}}{(\rtcformula{x}{y}{\varphi})(s, s)}$ by applying rule \eqref{eq:refRTC}, by \cref{lem:quotient} there must also be a cut-free proof of $\sequent{\Gamma_{\omega}}{(\rtcformula{x}{y}{\varphi})(s, t)}$.
      However, since $(\rtcformula{x}{y}{\varphi})(s, t) \in \Delta_{\omega}$ this would imply that $\sequent{\Gamma_{\omega}}{\Delta_{\omega}}$ is cut-free provable, which contradicts \cref{not-prov}.
      \item
      Otherwise, there are $a_0, \ldots, a_1 \in D$ ($n > 0$) such that $\rho_{\omega}(s) = a_0$, $\rho_{\omega}(t) = a_n$ and $M_{\omega}, \rho_{\omega} [x := a_{i}, y := a_{i+1}] \models \varphi$ for each $i < n$. However this directly contradicts the auxiliary result proved above.
    \end{itemize}
    In both cases, we have derived a contradiction, and so we then conclude that $M_{\omega}, \rho_{\omega} \not\models (\rtcformula{x}{y}{\varphi})(s, t)$ as required.
    \qed
  \end{itemize}
\end{proof}

The completeness result therefore follows since, by construction, a sequent $S$ is contained within its corresponding limit sequents.

\begin{restatable}[Completeness]{theorem}{infinitarycompleteness}
  \label{thm:ORTC:Cut-free-complete}
  {\ortc} is complete for standard semantics.
\end{restatable}
\begin{proof}
  Now given any sequent $S$, if some search tree $T_{\omega}$ contracted for $S$ is not an {\ortc} proof then it follows from \Cref{lem:LimitSequents} that $S$ is not valid ($M_\omega$ is a counter model for it). Thus if $S$ is valid, then $T_\omega$ is a recursive \ortc~ proof for it.
  \qed
\end{proof}

We obtain admissibility of cut as the search tree $T_\omega$ is cut-free.

\begin{corollary}[Cut admissibility]
  \label[corollary]{thm:ORTC:Cut-elim}
  Cut is admissible in {\ortc}.
\end{corollary}

\subsection{{\rtclogic} with Pairs}

To obtain the full inductive expressivity we must allow the formation of the transitive closure of not only binary relations, but any $2n$-ary relation. In \cite{AvronTC03} it was shown that taking such a $RTC^n$ operator for every $n$ (instead of just for $n=1$) results in a more expressive logic, namely one that captures all finitary first-order definable inductive definitions and relations. Nonetheless, from a proof theoretical point of view having  infinitely many such operators is suboptimal.  Thus, we here instead incorporate the notion of ordered pairs and use it to encode such operators. For example, writing $\pair{x}{y}$ for the application of the pairing function $\langle\rangle(x,y)$, the formula $(RTC^2_{x_1, x_2, y_1, y_2}\,\varphi)(s_1, s_2, t_1, t_2)$ can be encoded by: 
\begin{equation*}
  (\rtcformula{x}{y}{\exists x_1, x_2, y_1, y_2 \suchthat x = \pair{x_1}{x_2} \wedge y = \pair{y_1}{y_2} \wedge \varphi})(\pair{s_1}{s_2}, \pair{t_1}{t_2})
\end{equation*}
Accordingly, we may assume languages that explicitly contain a pairing function, providing that we (axiomatically) restrict to structures that interpret it as such (i.e.~the \emph{admissible} structures). 
For such languages we can consider two induced semantics: admissible standard semantics and admissible Henkin semantics, obtained by restricting the (first-order part of the) structures to be admissible.

The above proof systems are extended to capture ordered pairs as follows.

\begin{definition}
  For a signature containing at least one constant $c$, and a binary function symbol denoted by $\pairSym$, the proof systems {\rtcPairs}, {\ortcPairs}, and {\cortcPairs} are obtained from {\rtc}, {\ortc}, {\cortc} (respectively) by the addition of the following rules:
  \begin{flalign*}
%    \label{eq:Pairs:OneToOne}
    \qquad
      &
    \begin{gathered}
    \infer
      {\sequent
        {\Gamma}
        {x = u \wedge y = v, \Delta}}
      {\sequent
        {\Gamma}
        {\pair{x}{y} = \pair{u}{v}, \Delta}}
    \end{gathered}
       &
%    \label{eq:Pairs:Existence}
    \begin{gathered}
      \infer
        {\sequent
          {\Gamma, \pair{x}{y} = c}
          {\Delta}}
        {}
    \end{gathered}
      &
    \qquad
  \end{flalign*}
\end{definition}

The proofs of \cref{rtc_comp,thm:ORTC:Cut-free-complete} can easily be extended to obtain the following results for languages with a pairing function. For completeness, the key observation is that the model of the counter-interpretation is one in which every binary function is a pairing function. That is, the interpretation of any binary function is such that satisfies the standard pairing axioms. Therefore, the model of the counter-interpretation is an admissible structure.

\begin{theorem}[Soundness and Completeness of {\rtcPairs} and {\ortcPairs}] 
  \label{rtcpairs_comp}
  The proof systems {\rtcPairs} and {\ortcPairs} are both sound and complete for the admissible forms of Henkin and standard semantics, respectively.
\end{theorem}

\section{Relating the Finitary and Infinitary Proof Systems}
\label{sec:ProofSystemComparisons}

This section discusses the relation between the explicit and the cyclic system for {\TC}. In \cref{rtc-cortc} we show that the former is contained in the latter. The converse direction, which is much more subtle, is discussed in \cref{cortc-rtc}.

\subsection{Inclusion of {\rtc} in {\cortc}}
\label{rtc-cortc}

Provability in the explicit induction system implies provability in the cyclic system. The key property is that we can derive the explicit induction rule in the cyclic system, as shown in \Cref{fig:implicit-explicit}.

\begin{figure*}[t]
  \begin{equation*}
  \scalebox{0.95}{
    \kern -1.75cm
    \begin{prooftree}
      \prooftree
        \prooftree
          \prooftree
            \justifies
            {\sequent
              {\Gamma, {\psi \subst{\mbox{\replace{v}{x}}}}}
              {\Delta, {\psi \subst{\mbox{\replace{v}{x}}}}}}
            \using {\rulename{WL,WR,Ax}}
          \endprooftree
          \justifies
          \tikz \coordinate (left-edge) at (-1em,0) ;
          {\sequent
            {\Gamma, {\psi \subst{\mbox{\replace{v}{x}}}}, v=w}
            {\Delta, {\psi \subst{\mbox{\replace{w}{x}}}}}}
          \using {\rulename{${=}\text{L}_1$}}
        \endprooftree
        {\multiput(37.5,0)(0,5){8}{.}
        \raise 1.5cm
        \hbox to 4.5cm {
        \kern -5.875cm
        \prooftree
          \prooftree
            \begin{gathered}
              \tikz \coordinate (bud) at (0,-0.5em) ; \\
              {\sequent
                {\Gamma, {\psi \subst{\mbox{\replace{v}{x}}}}, (\rtcformula{x}{y}{\varphi})(v, w)}
                {\Delta, {\psi \subst{\mbox{\replace{w}{x}}}}}}
              \\[0.25em]
            \end{gathered}
            \justifies
            \begin{gathered}
            \vspace{0.125em}
            {\sequent
              {\Gamma, {\psi \subst{\mbox{\replace{v}{x}}}}, (\rtcformula{x}{y}{\varphi})(v, z)}
              {\Delta, {\psi \subst{\mbox{\replace{z}{x}}}}}}
            \end{gathered}
            \using {\rulename{Subst}}
          \endprooftree
          \prooftree
            \begin{gathered}
            \sequent
              {\Gamma, \psi, \varphi}
              {\Delta, {\psi \subst{\mbox{\replace{y}{x}}}}}
            \\[0.375em]
            \end{gathered}
            \justifies
            \begin{gathered}
            \sequent
              {\Gamma, {\psi \subst{\mbox{\replace{z}{x}}}}, {\varphi \subst{\scalebox{0.865}{\replace{z}{x}, \replace{w}{y}}}}}
              {\Delta, {\psi \subst{\mbox{\replace{w}{x}}}}}
            \\[0.3em]
            \end{gathered}
            \using {\rulename{Subst}}
          \endprooftree
          \justifies
          {\sequent
            {\Gamma, {\psi \subst{\mbox{\replace{v}{x}}}}, (\rtcformula{x}{y}{\varphi})(v, z), {\varphi \subst{\scalebox{0.865}{\replace{z}{x}, \replace{w}{y}}}}}
            {\Delta, {\psi \subst{\mbox{\replace{w}{x}}}}}}
          \using {\rulename{Cut}}
        \endprooftree}}
        \justifies
        \begin{gathered}
        {\tikz \coordinate (companion) at (-0.5em,0.25em) ;
         \sequent
          {\Gamma, {\psi \subst{\mbox{\replace{v}{x}}}}, (\rtcformula{x}{y}{\varphi})(v, w)}
          {\Delta, {\psi \subst{\mbox{\replace{w}{x}}}}}}
          \\[0.36em]
        \end{gathered}
        \using {\rulename{\ref{eq:RTC-cyc}}}
      \endprooftree
      \justifies
      {\sequent
        {\Gamma, {\psi \subst{\mbox{\replace{s}{x}}}}, (\rtcformula{x}{y}{\varphi})(s, t)}
        {\Delta, {\psi \subst{\mbox{\replace{t}{x}}}}}}
      \using {\rulename{Subst}}
    \end{prooftree}
    \begin{tikzpicture}[overlay]
      \draw[arrow] (bud) -- ++(0,0.5em) -| (left-edge) |- (companion) [->];
    \end{tikzpicture}}
  \end{equation*}
  \caption{{\cortc} derivation simulating Rule \eqref{eq:RTC-min}. The variables $v$ and $w$ are fresh (i.e.~do not occur free in $\Gamma$, $\Delta$, $\varphi$, or $\psi$).}
  \label{fig:implicit-explicit}
\end{figure*}

\begin{lemma}
  \label[lemma]{lem:RTC-min:derivable:in:cortc}
  Rule (\ref{eq:RTC-min}) is derivable in $\cortc$.
\end{lemma}

This leads to the following result (an analogue to \cite[Thm.~7.6]{brotherston2010sequent}).

\begin{restatable}{theorem}{easyinclusion}
  \label{easyinclusion}
  $\cortc \supseteq \rtc$, and is thus complete w.r.t. Henkin semantics.
\end{restatable}
\begin{proof}
  Let $\mathcal{P}$ be a proof in {\rtc} and $\mathcal{P}'$ be the corresponding pre-proof in {\cortc} obtained be replaing each instance of Rule \eqref{eq:RTC-min} by the corresponding instance of the proof schema given in \cref{lem:RTC-min:derivable:in:cortc}. We argue that $\mathcal{P}'$ is a valid {\cortc} proof. Notice that the only cycles in $\mathcal{P}'$ are internal to the subproofs that simulate Rule \eqref{eq:RTC-min}. Thus any infinite path in $\mathcal{P}'$ must eventually end up traversing one of these cycles infinitely often. Therefore, it suffices to show that there is an infinitely progressing trace following each such path. This is clearly the case since we can trace the active {\rtcOperator} formulas along these paths, which progress once each time around the cycle, across Rule \eqref{eq:RTC-cyc}.
  \qed
\end{proof}

\Cref{lem:RTC-min:derivable:in:cortc} is the {\TC} counterpart of \cite[Lemma~7.5]{brotherston2010sequent}. It is interesting to note that the simulation of the explicit {\lkid} induction rule in the cyclic {\lkid} system is rather complex since each predicate has a slightly different explicit induction rule, which depends on the particular productions defining it. Thus, the construction for the cyclic {\lkid} system must take into account the possible forms of arbitrary productions. In contrast, {\cortc} provides a single, uniform way to unfold an {\rtcOperator} formula: the construction given in \cref{fig:implicit-explicit} is the cyclic representation of the {\rtcOperator} operator semantics, with the variables $v$ and $w$ implicitly standing for arbitrary terms (that we subsequently substitute for). 

This uniform syntactic translation of the explicit {\rtc} induction rule into {\cortc} allows us to \emph{syntactically} identify a proper subset of cyclic proofs which is also complete w.r.t. Henkin semantics.\footnote{Note it is not clear that a similar complete structural restriction is possible for {\lkid}.}
The criterion we use is based on the notion of \emph{overlapping cycles}. Recall the definition of a \emph{basic} cycle, which is a path in a (proof) graph starting and ending at the same point, but containing no other repeated nodes. We say that two distinct (i.e.~not identical up to permutation) basic cycles \emph{overlap} if they share any nodes in common, i.e.~at some point they both traverse the same path in the graph. We say that a cyclic proof is \emph{non-overlapping} whenever no two distinct basic cycles it contains overlap. The restriction to non-overlapping proofs has an advantage for automation, since one has only to search for cycles in one single branch.

\begin{definition}[Normal Cyclic Proofs]
  The \emph{normal} cyclic proof system {\ncortc} is the subsystem of {\ortc} comprising of all and only the non-overlapping cyclic proofs.
\end{definition}

The following theorem is immediate due to the fact that the translation of an $\rtc$ proof into $\cortc$, using the construction shown in \Cref{fig:implicit-explicit}, results in a proof with no overlapping cycles.

\begin{theorem}
  \label{normal}
  $\ncortc \supseteq \rtc$.
\end{theorem}

Henkin-completeness of the normal cyclic system then follows from \Cref{normal} and \Cref{rtc_comp}.

\subsection{Inclusions of {\cortc} in {\rtc}}
\label{cortc-rtc}

This section addresses the question of whether the cyclic system is equivalent to the explicit one, or strictly stronger. In \cite{brotherston2010sequent} it was conjectured that for the system with inductive definitions,  $\lkid$ and $\colkid$ are equivalent. Later, it was shown that they are indeed equivalent when containing arithmetics \cite{Berardi2017b,Simpson2017}. We obtain a corresponding theorem in \Cref{arithmetics} for the {\TC} systems. However, it was also shown in \cite{Berardi2017} that in the general case the cyclic system is stronger than the explicit one. We discuss the general case for {\TC} and its subtleties in \Cref{hydra}. 

\subsubsection{The Case of Arithmetics}
\label{arithmetics}

Let $\rtclogic$ be a language based on the signature %that contains 
$\{0,\fn{s},{+}\}$. Let $\rtc\plusArith$ and $\cortc\plusArith$ be the systems for $\rtclogic$ obtained by adding to {\rtc} and {\cortc}, respectively, the standard axioms of {\PA} together with the $\rtcOperator$-characterization of the natural numbers, i.e.: 
\begin{enumerate}[nosep,leftmargin=3em,label={\roman*)},ref={\roman*}]
  \item
  $\sequent{ \suc{x} = 0}{}$
  \item
  $\sequent{\suc{x} = \suc{y}}{x = y}$
  \item
  $\sequent{}{x + 0 = x}$
  \item
  $\sequent{}{x + \suc{y} = \suc{(x + y)}}$
  \item
  \label{ax:everything-is-nat}
  $\sequent{}{(\rtcformula{w}{u}{\suc{w} = u})(0,x)}$
\end{enumerate}
Note that we do not need to assume multiplication explicitly in the signature, nor do we need to add axioms for it, since multiplication is definable in $\rtclogic$ and its standard axioms are derivable~\cite{AvronTC03,cohen2015middle}. 

Recall that we can express facts about sequences of numbers in {\PA} by using a $\beta$-function such that for any finite sequence $k_{0},k_{1},...,k_{n}$ there is some $c$ such that for all $i \leq n$, $\beta(c,i) = k_{i}$.
%Thus, our motivation is that $\left(RTC_{x,y}\varphi\right)\left(s,t\right)$ holds iff $s=t$ or for some $n$, there is a sequence $k_{0},k_{1},...,k_{n}$ such that $k_{0}=I\left[s\right]$, $k_{n}=I\left[t\right]$, and each pair of consecutive terms are in the relation defined by $\varphi$.
Accordingly, let $B$ be a well-formed formula of the language of {\PA} with three free variables which captures in {\PA} a $\beta$-function. For each formula $\varphi$ of the language of {\PA} define $\varphi^{\beta} := \varphi$, and define $((\rtcformula{x}{y}{\varphi})(s, t))^{\beta}$ to be:
\begin{multline*}
  s = t 
    \vee 
  (\exists z, c \suchthat 
    B(c,0,s) \wedge B(c,\suc{z},t) 
      \wedge 
        {} \\ 
    (\forall u \leq z \holdsthat 
      \exists v, w \suchthat 
        B(c,u,v) \wedge B(c,\suc{u},w) \wedge {\varphi^{\beta} \subst{\scalebox{0.8}{\replace{v}{x}, \replace{w}{y}}}}))
\end{multline*}

The following result, which was proven in \cite{CohenMSc,cohen2015middle}, establishes an equivalence between $\rtc\plusArith$ and {\pa} (a Gentzen-style system for {\PA}). It is mainly based on the fact that in $\rtc\plusArith$ all instances of {\pa} induction rule are derivable.

\begin{theorem}[cf.~\cite{cohen2015middle}] %
  \label{lem:PA-RTC} %
  The following hold:
  \begin{enumerate}[nosep]
    \item \label{lem:PA-RTC:formula}
    $\vdash_{\rtc\plusArith} \varphi \Leftrightarrow \varphi^{\beta}$. 
    \item \label{lem:PA-RTC:sequent}
    $\vdash_{\rtc\plusArith} \sequent{\Gamma}{\Delta}$ iff ${} \vdash_{\pa} \sequent{\Gamma^{\beta}}{\Delta^{\beta}}$.
  \end{enumerate}
\end{theorem}

We show a similar equivalence holds between the cyclic system {\cortc} and {\ca}, a cyclic system for arithmetic shown to be equivalent to {\pa} \cite{Simpson2017}.
The first part of the result is straightforward.

\begin{lemma}
  \label{lem:CRTC+A:formulas}
  $\vdash_{\cortc\plusArith} \varphi \Leftrightarrow \varphi^{\beta}$
\end{lemma}
\begin{proof}
  This follows from \Cref{lem:PA-RTC}\eqref{lem:PA-RTC:formula} and \Cref{easyinclusion}.
  \qed
\end{proof}

To show the second part, we first show that {\ca} is included in ${\cortc\plusArith}$ by giving a construction that directly translates {\ca} proofs into ${\cortc\plusArith}$ proofs. Technically, the signature of {\ca} includes the relation symbol $<$ for strict ordering, and the function symbol $\cdot$ for multiplication. As mentioned above, multiplication (and its axioms) are derivable in {\TC}, and the non-strict ordering on natural numbers $s \leq t$ can be expressed in {\TC} as $(\rtcformula{w}{u}{\suc{w} = u})(s, t)$. Therefore, in the following result, we implicitly assume that all {\ca} terms of the form $s \cdot t$ are translated in {\cortc} as in \cite{AvronTC03,cohen2015middle}, and formulas of the form $s < t$ as $s \neq t \wedge (\rtcformula{w}{u}{\suc{w} = u})(s, t)$.

\begin{lemma}
  \label{lem:CA-in-CRTC+A}
  If ${} \vdash_{\ca} \sequent{\Gamma}{\Delta}$ then ${} \vdash_{\cortc\plusArith} \sequent{\Gamma}{\Delta}$.
\end{lemma}
\begin{proof}
  We define a translation $[\cdot]^{\ast}$ on sets of formulas $\Gamma$ as the smallest set of formulas containing $\Gamma$ and $(\rtcformula{w}{u}{\suc{w} = u})(0, t)$ for each term $t$ that appears as a subterm of a formula in $\Gamma$ with no quantified variables (we call such a term \emph{free}). Notice it suffices to prove ${} \vdash_{\cortc\plusArith} \sequent{\Gamma^{\ast}}{\Delta}$ if ${} \vdash_{\ca} \sequent{\Gamma}{\Delta}$, since then we may cut all the added formulas $(\rtcformula{w}{u}{\suc{w} = u})(0, t)$ in $\Gamma^{\ast}$ using instances of axiom \eqref{ax:everything-is-nat}.
  
  We call $\sequent{\Gamma^{\ast}}{\Delta}$ the $\ast$-translation of the sequent $\sequent{\Gamma}{\Delta}$, and call the inference rule obtained by applying this translation to the conclusion and each premise the $\ast$-translation of the rule. We begin by showing that the $\ast$-translations of the axioms and proof rules of {\ca} are derivable in {\cortc\plusArith}. For the standard rules of {\LK} and the substitution rule this is trivial, and can be done is such a way that, for the conclusion $\sequent{\Gamma}{\Delta}$ and each premise $\sequent{\Gamma'}{\Delta'}$, there is a (non-progressing) trace from each $(\rtcformula{w}{u}{\suc{w} = u})(0, t) \in \Gamma^{\ast}$ to $(\rtcformula{w}{u}{\suc{w} = u})(0, t') \in (\Gamma')^{\ast}$ when $t'$ is a predecessor of $t$ (cf.~\cite[Def.~1]{Simpson2017}). It remains to show that we can derive the $\ast$-translations of the following set of axioms:
  \begin{flalign*}
    \hspace{1.5cm}
      &
    \begin{aligned}
      t < u, u < v &\Rightarrow t < v \\
      t < u, u < t &\Rightarrow \\
      t < u, u < \suc{t} &\Rightarrow \\
      t < 0 &\Rightarrow \\
      t < u &\Rightarrow \suc{t} < \suc{u} \\
      &\Rightarrow t < \suc{t}
    \end{aligned}
      &
    \begin{aligned}
      &\Rightarrow t < u, t = u, u < t \\
      &\Rightarrow t + 0 = t \\
      &\Rightarrow t + \suc{u} = \suc{(t + u)} \\
      &\Rightarrow t \cdot 0 = 0 \\
      &\Rightarrow t \cdot \suc{u} = (t \cdot u) + t
    \end{aligned}
      &
    \hspace{1.5cm}
  \end{flalign*}
  which is straightforward. It is also similarly straightforward to show that the $\ast$-translation of the following inference rule:
  \begin{align*}
    \begin{gathered}
    \infer
      [({\text{$x$ is fresh}})]
      {\sequent{\Gamma, 0 < t}{\Delta}}
      {\sequent{\Gamma, t = \suc{x}}{\Delta}}
    \end{gathered}
  \end{align*}
  is derivable in {\cortc\plusArith} such that there is a (non-progressing) trace from each $(\rtcformula{w}{u}{\suc{w} = u})(0, t) \in \Gamma^{\ast}$ in the conclusion to the corresponding formula in the premise of the derivation.
  \rrnote{Liron, I'm not sure how much/little of these derivations you want to include here.}%
  \lcnote{I can do one or two as examples, but I think this is straightforward. }%
  We shall call these derivations of the $\ast$-translations of the {\ca} axioms ad rules the \emph{simple} derivations.  
  
  We next show that for each (non-axiomatic) rule of {\ca}, we can derive the $\ast$-translation of the rule in {\cortc\plusArith} with the additional property that for every \emph{progressing} trace from a term $t$ in the conclusion to a term $t'$ in a premise, there is also a progressing trace from the formula $(\rtcformula{w}{u}{u = \suc{w}})(0, t)$ in the conclusion of the derived rule to the formula $(\rtcformula{w}{u}{u = \suc{w}})(0, t')$ in the corresponding premise.
  
  \begin{figure*}[t]
    \resizebox{\textwidth}{!}{\begin{minipage}{1.15\textwidth}
    \begin{gather*}
      \kern-1cm
      \infer[{\eqref{eq:RTC-cyc}}]
        {\sequent{\Gamma, t' < t, \pred{N}{t'}, \pred{N}{t}}{\Delta}}
        {\infer[(\text{WL/WR})]
          {\sequent{\Gamma, t' < 0, \pred{N}{t'}, \pred{N}{0}}{\Delta}}
          {\infer*
            {\sequent{t' < 0, \pred{N}{t'}, \pred{N}{0}}{}}
            {\dag}}
          &
         {\hbox to 6cm {
          \infer[{({=}\text{L}')}]
           {\sequent{\Gamma, t' < t, \pred{N}{t'}, \pred{N}{t}, \pred{N}{z}, t = \suc{z}}{\Delta}}
           {\multiput(50,0)(0,5){5}{.}
            \raisebox{1cm}{
            \hbox to 5cm {
            \kern-5cm
            \infer[(\text{Cut})]
             {\sequent{\Gamma, t' < t, t' < \suc{z}, \pred{N}{t'}, \pred{N}{t}, \pred{N}{\suc{z}}, \pred{N}{z}}{\Delta}}
             {\multiput(50,0)(0,5){34}{.}
              \raisebox{6cm}{
              \hbox to 2cm {
              \kern -2.3cm
              \infer[(\text{$\vee$R})]
               {\sequent{\pred{N}{t'}, \pred{N}{\suc{z}}, t' < \suc{z}}{t' = z \vee t' < z}}
               {\infer[(\text{Cut})]
                 {\sequent{\pred{N}{t'}, \pred{N}{\suc{z}}, t' < \suc{z}}{t' = z, t' < z}}
                 {\infer*{\sequent{}{z < t', t' = z, t' < z}}{\ddag}
                   &
                  \infer*{\sequent{\pred{N}{t'}, \pred{N}{\suc{z}}, z < t', t' < \suc{z}}{}}{\maltese}}}}}
               &
              \infer[(\text{$\vee$L})]
                {\sequent{\Gamma, t' < t, t' = z \vee t' < z, \pred{N}{t'}, \pred{N}{t}, \pred{N}{z}}{\Delta}}
                {\infer[{(\text{WL})}]
                  {\sequent{\Gamma, t' < t, t' = z, \pred{N}{t'}, \pred{N}{t}, \pred{N}{z}}{\Delta}}
                  {{\ast} \quad \sequent{\Gamma, t' < t, \pred{N}{t'}, \pred{N}{t}}{\Delta}}
                  &
                 \multiput(3,0)(0,5){8}{.}
                 \raisebox{1.5cm}{
                 \hbox to 3cm {
                 \kern-6cm\infer[{\eqref{eq:RTC-cyc}}]
                   {\sequent{\Gamma, t' < t, t' < z, \pred{N}{t'}, \pred{N}{t}, \pred{N}{z}}{\Delta}}
                   {\multiput(10,0)(0,5){10}{.}
                    \raisebox{1.75cm}{
                    \hbox to 0.5cm {
                    \kern-1cm
                    \infer[(\text{WL/WR})]
                     {\sequent{\Gamma, t' < t, t' < 0, \pred{N}{t'}, \pred{N}{t}, \pred{N}{0}}{\Delta}}
                     {\infer*
                       {\sequent{t' < 0, \pred{N}{t'}, \pred{N}{0}}{}}
                       {\dag}}}}
                     &
                    \infer[{({=}\text{L}')}]
                     {\sequent{\Gamma, t' < t, t' < z, \pred{N}{t'}, \pred{N}{t}, \pred{N}{z}, \pred{N}{z'}, z = \suc{z'}}{\Delta}}
                     {\infer[(\text{Subst})]
                       {\sequent{\Gamma, t' < t, t' < \suc{z'}, \pred{N}{t'}, \pred{N}{t}, \pred{N}{\suc{z'}}, \pred{N}{z'}}{\Delta}}
                       {\sequent{\Gamma, t' < t, t' < \suc{z}, \pred{N}{t'}, \pred{N}{t}, \pred{N}{\suc{z}}, \pred{N}{z}}{\Delta}
                        \tikz \coordinate (bud) at (0,0);}}}}}}}}}}}}}
      \begin{tikzpicture}[overlay]
        \draw[arrow] (bud) ++(-1.5em,1em) -- ++(0,1em) -- ++(2.1cm,0) -- ++(0,-5cm) |- ++(-4.5cm,0) -| ++(0,1em) [->];
      \end{tikzpicture}
    \end{gather*}
    \end{minipage}}
    \caption{A derivation schema simulating a {\ca} trace progression point in {\cortc\plusArith}.}
    \label{fig:CA-in-CRTC+A:TraceAwareDerivation}
  \end{figure*}  
  
  Consider the open derivation schema shown in \Cref{fig:CA-in-CRTC+A:TraceAwareDerivation}. Here, $\pred{N}{t}$ abbreviates the formula $(\rtcformula{w}{u}{u = \suc{w}})(0, t)$, and $t' < t$ abbreviates the translation given above, i.e.~$s \neq t \wedge (\rtcformula{w}{u}{\suc{w} = u})(s, t)$. 
  The symbols $\dag$, $\ddag$ and $\maltese$ denote the (simple) derivations of the $\ast$-translations of the appropriate axioms.
  We also write (${=}\text{L}'$) in \Cref{fig:CA-in-CRTC+A:TraceAwareDerivation} to refer to instances of the following general schema for a derived equality rule, in which $\Sigma(t) = \{ \varphi_{1}(t), \ldots, \varphi_{n}(t) \}$ and $\Sigma(t)$ may or may not occur in the right-hand premise of the (Cut) rule instance.
  \begin{gather*}
    \scalebox{0.8}{
    \begin{prooftree}
      \prooftree
        \prooftree
          \prooftree
            \prooftree
              \justifies
              \sequent{\varphi_{1}(t)}{\varphi_{1}(t)}
              \using {\text{\smaller{(Ax)}}}
            \endprooftree
            \justifies
            \sequent{\Sigma(t)}{\varphi_{1}(t)}
            \using {\text{\smaller{(WL)}}}
          \endprooftree
            \cdots
          \prooftree
            \prooftree
              \justifies
              \sequent{\varphi_{n}(t)}{\varphi_{n}(t)}
              \using {\text{\smaller{(Ax)}}}
            \endprooftree
            \justifies
            \sequent{\Sigma(t)}{\varphi_{n}(t)}
            \using {\text{\smaller{(WL)}}}
          \endprooftree
          \justifies
          \sequent{\Sigma(t)}{\wedge_{i \leq n} \, \varphi_{i}(t)}
          \using {\text{\smaller{($\wedge$R)}}}
        \endprooftree
        \justifies
        \sequent{\Sigma(t), t = u}{\wedge_{i \leq n} \, \varphi_{i}(u)}
        \using {\text{\smaller{(${=}\text{L}_1$)}}}
      \endprooftree
      \prooftree
        \sequent{\Gamma, [\Sigma(t)], \Sigma(u)}{\Delta}
        \justifies
        \sequent{\Gamma, [\Sigma(t)], {\wedge_{i \leq n} \, \varphi_{i}(u)}}{\Delta}
        \using {\text{\smaller{($\wedge$L)}}}
      \endprooftree
      \justifies
      \sequent{\Gamma, \Sigma(t), t = u}{\Delta}
      \using {\text{\smaller{(Cut)}}}
    \end{prooftree}
    }
  \end{gather*}
  Notice that, crucially, in the derivation in \Cref{fig:CA-in-CRTC+A:TraceAwareDerivation} there is a \emph{progressing} trace from $\pred{N}{t}$ in the conclusion to $\pred{N}{t'}$ in the premise (marked with a $\ast$). There is also a non-progressing trace from each {\rtcOperator} formula in $\Gamma$ in the conclusion to its occurrence in $\Gamma$ in the premise.% Thus it contains extra information which we will use to simulate the traces that appear in {\ca} proofs.
  
  For each {\ca} inference rule concluding $\sequent{\Gamma}{\Delta}$, we build a \emph{trace-aware} derivation of its $\ast$-translation as follows. First, we take the simple derivation of the inference rule. Notice that this gives a non-progressing trace from each $\pred{N}{t} \in \Gamma^{\ast}$ to $\pred{N}{t'}$ in a premise where $t'$ is a precursor of $t$ in the {\ca} inference rule.
  Now, for every premise $\sequent{\Gamma_{j}}{\Delta_{j}}$ of the rule, let $\{ t'_{1} < t_{1}, \ldots, t'_{n} < t_{n} \}$ be the set of all such formulas in $\Gamma_{j}$. Notice that $t'_{i} < t_{i}, \pred{N}{t'_{i}}, \pred{N}{t_{i}} \subseteq (\Gamma_{j})^{\ast}$ for each $i$. To each corresponding premise in the simple derivation we apply, in turn, $n$ instances of the derivation schema in \Cref{fig:CA-in-CRTC+A:TraceAwareDerivation}, one for each formula $t_{i} < t'_{i}$. Notice that this combined derivation satisfies the following for each premise $\sequent{\Gamma_{j}}{\Delta_{j}}$ and free terms $t \in \Gamma$ and $t_{i} \in \Gamma_{j}$:
  \begin{enumerate}[nosep]
    \item
    there is a non-progressing trace from $\pred{N}{t} \in \Gamma^{\ast}$ in the conclusion to $\pred{N}{t_{i}} \in (\Gamma_{j})^{\ast}$ in the premise if $t_{i}$ is a precursor of $t$; and
    \item
    there is a progressing trace from $\pred{N}{t} \in \Gamma^{\ast}$ in the conclusion to $\pred{N}{t'_{i}} \in (\Gamma_{j})^{\ast}$ in the premise if $t_{i}$ is a precursor of $t$ and $t'_{i} < t_{i} \in \Gamma_{j}$.
  \end{enumerate}
  Note that this means the notion of trace in the {\ca} inference rules is exactly mirrored by a {\cortc} trace in the trace-aware derivation.
  Moreover, the properties above hold of \emph{all} paths from the conclusion to a premise in the trace-aware derivation, in particular those that travel around the internal cycles any finite number of times.
  
  From this construction it follows that we can transform a {\ca} pre-proof, via the trace-aware local $\ast$-translation of each rule, into a {\cortc\plusArith} pre-proof with the same global structure. %, such that for every trace following a path in the {\ca} proof there is a trace with the same number of progression points following the corresponding path in the {\cortc\plusArith} pre-proof.
  It remains to show that each such {\cortc\plusArith} pre-proof resulting from a {\ca} proof is also a {\cortc\plusArith} proof. That is, it satisfies the {\cortc} global trace condition. Consider an arbitrary infinite path in the {\cortc\plusArith} pre-proof. There are two cases to consider:
  \begin{itemize}[nosep]
    \item
    The infinite path ends up traversing an infinite path local to the (trace-aware) $\ast$-translation of an inference rule or {\ca} axiom; in this case notice that each such infinite path has an infinitely progressing trace.
    \item
    The infinite path corresponds to an infinite path in the {\ca} proof (possibly interspersed with finite traversals of the cycles local to the trace-aware \nohyphens{$\ast$-translation} of each rule instance). Since there is an infinitely progressing trace following the path in the {\ca} proof, by the properties above there is also a corresponding infinitely progressing trace following the path in the {\cortc\plusArith} pre-proof.
    \qed
  \end{itemize}
\end{proof}

This leads immediately to the `if' direction.

\begin{corollary}
  \label{cor:CA-Beta-in-CRTC+A}
  If ${} \vdash_{\ca} \sequent{\Gamma^{\beta}}{\Delta^{\beta}}$ then ${} \vdash_{\cortc\plusArith} \sequent{\Gamma}{\Delta}$.
\end{corollary}
\begin{proof}
  We first use \Cref{lem:CA-in-CRTC+A} to derive $\sequent{\Gamma^{\beta}}{\Delta^{\beta}}$ in $\cortc\plusArith$, and then combine this with derivations in $\cortc\plusArith$ of $\sequent{\varphi}{\varphi^{\beta}}$ (resp.~$\sequent{\varphi^{\beta}}{\varphi}$) for each $\varphi \in \Gamma$ (resp.~$\varphi \in \Delta$), which exist by \Cref{lem:CRTC+A:formulas}, with applications of cuts to derive $\sequent{\Gamma}{\Delta}$.
  \qed
\end{proof}

%The translation of {\rtcOperator} formulas via the $\beta$-function with the existentially quantified variable $z$ is analogous to the construction in 
In \cite{Berardi2017b}, to show the equivalence of the explicit and cyclic systems for {\lkid}, a construction was given which translates {\lkid} predicates $P(\vec{t})$ into predicates $P'(\vec{t}, n)$ with equivalent inductive definitions and an extra parameter $n$ comprising a `stage' variable. The equivalence is derived by using the cycles in a proof to construct an explicit induction hypothesis over these stage variables. Here, for {\TC}, instead of directly constructing an induction hypothesis for the explicit system, we show that from a $\cortc\plusArith$ proof we can construct an analogous proof in {\ca} which preserves cycles, and then use the existing equivalence results between {\ca}, {\pa} and {\rtc}. Our construction is similar to the one given in \cite{Berardi2017b}, in that we use a variant of the $\beta$-function which introduces a free variable $n$ (similar to a stage variable), which we are able to trace in the cyclic {\ca} proof. This results in the `only if' direction of the result.

\begin{lemma}
  \label{lem:CRTC+A-Beta-in-CA}
  If ${} \vdash_{\cortc\plusArith} \sequent{\Gamma}{\Delta}$ then ${} \vdash_{\ca} \sequent{\Gamma^{\beta}}{\Delta^{\beta}}$.
\end{lemma}
\begin{proof}
  Similarly to the proof of \cref{lem:CA-in-CRTC+A} above, we define a local translation on proof rules that preserves {\cortc\plusArith} traces as {\ca} traces. For this, we use a \emph{parameterised} variant $\bar{\beta}[n]$ of the $\beta$-translation, which introduces its parameter as a free variable in the translation. It is defined in the same way as the $\beta$-translation, except that it translates formulas of the form $(\rtcformula{x}{y}{\varphi})(s, t)$ as follows:
  \begin{multline*}
    s = t \vee \exists z, c \suchthat n = \suc{z} \wedge B(c,0,s) \wedge B(c,\suc{z},t) \wedge {}
      \\
    \forall u \leq z \holdsthat \exists v, w \suchthat B(c,u,v) \wedge B(c,\suc{u},w) \wedge {\varphi^{\beta} \subst{\scalebox{0.865}{\replace{v}{x}, \replace{w}{y}}}}
  \end{multline*}
  Notice the use of the original $\beta$-translation for the body of the {\rtcOperator} formula $\varphi$.
  We extend the $\bar{\beta}$-translation to sets of formulas, sequents, and inference rules as follows:
  \begin{itemize}[wide,itemsep=0.5em]
    \item
    for a set of formulas $\Gamma$, we define $\Gamma^{\bar{\beta}}$ as the set of $\bar{\beta}$-translations of the formulas in $\Gamma$ such that each translation of an $\rtcOperator$ sub-formula introduces a fresh free variable $z$; that is, each distinct {\rtcOperator} sub-formula is translated using a distinct variable parameter;
    \item
    for sequents, we define $(\sequent{\Gamma}{\Delta})^{\bar{\beta}} = \sequent{\Gamma^{\bar{\beta}}}{\Delta^\beta}$ such that the free variable parameters used by the $\bar{\beta}$-translation of the antecedent $\Gamma$ are distinct from the free variables in the succedent $\Delta$;
    \item
    for an inference rule with premises $\sequent{\Gamma_1}{\Delta_1}, \ldots, \sequent{\Gamma_n}{\Delta_n}$ and conclusion $\sequent{\Gamma}{\Delta}$, we define its $\bar{\beta}$-translation as the inference rule with premises $(\sequent{\Gamma_1}{\Delta_1})^{\bar{\beta}}, \ldots, (\sequent{\Gamma_n}{\Delta_n})^{\bar{\beta}}$ and conclusion $(\sequent{\Gamma}{\Delta})^{\bar{\beta}}$ such that multiple occurrences of the same {\rtcOperator} sub-formula across the original premises and conclusion are translated using the same free variable parameter in each of the translated premises and conclusion.
  \end{itemize}

  We show that if ${} \vdash_{\cortc\plusArith} \sequent{\Gamma}{\Delta}$ then ${} \vdash_{\ca} (\sequent{\Gamma}{\Delta})^{\bar{\beta}}$.
  We first prove that the $\bar{\beta}$-translation of each {\cortc\plusArith} inference rule can be derived in {\ca} in such a way that there is a {\ca} trace (progressing or non-progressing, as appropriate) simulating each {\cortc\plusArith} trace present in the original rule, which we do by tracing the free variable parameters in the {\ca} rule. We show how this is done for Rule \eqref{eq:RTC-cyc}; the other rules are straightforward, and do not contain progressing traces.
 
  \begin{figure*}[!t]
    \subfloat[One step from $s$ to $t$.]{
    \label{fig:CRTC+A-Beta-in-CA:TraceAwareDerivation:OneStep}
    \resizebox{\textwidth}{!}{\begin{minipage}{1.15\textwidth}
      \begin{gather*}
        \kern-1.5cm
        \infer[(\text{$\exists$R/$\wedge$R})]
          {\sequent{z = 0, \Sigma(z)}{\exists z, m \suchthat m < n \wedge (s = z \vee A(m, s, z)) \wedge {\varphi^{\beta} \subst{\scalebox{0.865}{\replace{z}{x},\replace{t}{y}}}}}}
          {\infer[(\text{$=$L})]
            {\sequent{z = 0, \Sigma(z)}{0 < n}}
            {\infer[(\text{{\sf{PA}}-Ax})]{\sequent{z = 0, \Sigma(z)}{0 < \suc{0}}}{}}
              &
           \multiput(10,0)(0,5){8}{.}
           \raisebox{1.5cm}{
           \hbox to 3.5cm {
           \kern -5cm
             \infer[(\text{$\vee$R})]
              {\sequent{z = 0, \Sigma(z)}{s = s \vee A(0, s, s)}}
              {\infer[(\text{$=$R})]{\sequent{z = 0, \Sigma(z)}{s = s}}{}}
           }}
              &
           \multiput(-40,0)(0,5){16}{.}
           \raisebox{3cm}{
           \hbox to 2.5cm {
           \kern-7cm
             \infer[(\text{WL/$=$L})]
               {\sequent{z = 0, \Sigma(z)}{\varphi^{\beta} \subst{\scalebox{0.865}{\replace{z}{x},\replace{t}{y}}}}}
               {\infer[(\text{$\leq_{0}$/$\exists$L})]
                 {\sequent{B(c, 0, s), B(c, \suc{0}, t), \forall u \leq 0 \holdsthat \vartheta(u)}{\varphi^{\beta} \subst{\scalebox{0.865}{\replace{z}{x},\replace{t}{y}}}}}
                 {\infer[(B)]
                   {\sequent{B(c, 0, s), B(c, \suc{0}, t), B(c, 0, v), B(c, \suc{0}, w), {\varphi^{\beta} \subst{\scalebox{0.865}{\replace{v}{x},\replace{w}{y}}}}}{\varphi^{\beta} \subst{\scalebox{0.865}{\replace{z}{x},\replace{t}{y}}}}}
                   {\infer[(\text{$=$L})]
                     {\sequent{s = v, t = w, \varphi^{\beta} \subst{\scalebox{0.865}{\replace{v}{x},\replace{w}{y}}}}{\varphi^{\beta} \subst{\scalebox{0.865}{\replace{s}{x},\replace{t}{y}}}}}
                     {\infer[(\text{Ax})]
                       {\sequent
                         {\varphi^{\beta} \subst{\scalebox{0.865}{\replace{s}{x},\replace{t}{y}}}}
                         {\varphi^{\beta} \subst{\scalebox{0.865}{\replace{s}{x},\replace{t}{y}}}}
                       }{}}
                   }
                 }
               }
             }}
          }
      \end{gather*}
    \end{minipage}}}
      \\[1em]
    \subfloat[Multi-step from $s$ to $t$.]{
    \label{fig:CRTC+A-Beta-in-CA:TraceAwareDerivation:MultiStep}
    \resizebox{\textwidth}{!}{\begin{minipage}{1.35\textwidth}
      \begin{gather*}
        \kern-0.5cm
         \infer[(\text{$\exists$L/$=$L})]
          {\sequent{\exists z' \suchthat z = \suc{z'}, \Sigma(z)}{\psi}}
          {\infer[(\leq_{\suc{}})]
             {\sequent{n = \suc{\suc{z'}}, B(c, 0, s), B(c, \suc{\suc{z'}}, t), \forall u \leq \suc{z'} \holdsthat \vartheta(u)}{\psi}}
               {\infer[(\text{$\exists$L/$\wedge$L})]
                 {\sequent{n = \suc{\suc{z'}}, B(c, 0, s), B(c, \suc{\suc{z'}}, t), \vartheta(\suc{z'}), \forall u \leq z' \holdsthat \vartheta(u)}{\psi}}
                 {\infer[(\text{$\exists$R/$\wedge$R})]
                   {\sequent{\Pi(z')}{\exists z, m \suchthat m < n \wedge (s = z \vee A(m, s, z)) \wedge {\varphi^{\beta} \subst{\scalebox{0.865}{\replace{z}{x},\replace{t}{y}}}}}}
                   {\infer[(\text{$=$L})]
                     {\sequent{\Pi(z')}{\suc{z'} < n}}
                     {\infer[(\text{{\sf{PA}}-Ax})]{\sequent{\Pi(z')}{\suc{z'} < \suc{\suc{z'}}}}{}}
                       &
                      \multiput(10,0)(0,7){8}{.}
                      \raisebox{2cm}{
                      \hbox to 1cm {
                      \kern -7cm
                        \infer[(\text{$\vee$R})]
                          {\sequent{\Pi(z')}{s = v \vee A(\suc{z'}, s, v)}}
                          {\infer[(\text{$\exists$R/$\wedge$R})]
                            {\sequent{\Pi(z')}{\exists z, c \suchthat \suc{z'} = \suc{z} \wedge B(c,0,s) \wedge B(c,\suc{z},v) \wedge \forall u \leq z \holdsthat \vartheta(u)}}
                            {\infer[(\text{$=$R})]{\sequent{\Pi(z')}{\suc{z'} = \suc{z'}}}{}
                              &
                             \infer[(\text{Ax})]{\sequent{\Pi(z')}{B(c, 0, s)}}{}
                              &
                             \infer[(\text{Ax})]{\sequent{\Pi(z')}{B(c, \suc{z'}, v)}}{}
                              &
                             \infer[(\text{Ax})]{\sequent{\Pi(z')}{\forall u \leq z' \holdsthat \vartheta(u)}}{}}}
                       }}
                         &
                       \infer[(B/WL)]
                         {\sequent{\Pi(z')}{\varphi^{\beta} \subst{\scalebox{0.865}{\replace{v}{x},\replace{t}{y}}}}}
                         {\infer[(\text{$=$L})]
                           {\sequent{t = w, \varphi^{\beta} \subst{\scalebox{0.865}{\replace{v}{x},\replace{w}{y}}}}{\varphi^{\beta} \subst{\scalebox{0.865}{\replace{v}{x},\replace{t}{y}}}}}
                           {\infer[(\text{Ax})]
                             {\sequent{\varphi^{\beta} \subst{\scalebox{0.865}{\replace{v}{x},\replace{t}{y}}}}{\varphi^{\beta} \subst{\scalebox{0.865}{\replace{v}{x},\replace{t}{y}}}}}{}}}}}}
          }
      \end{gather*}
    \end{minipage}}}
      \\[1em]
    \resizebox{\textwidth}{!}{\begin{minipage}{1.15\textwidth}
      \begin{gather*}
        \infer[(\text{$\exists$L})]
          {\sequent{A(n, s, t)}{\exists z, m \suchthat m < n \wedge (s = z \vee A(m, s, z)) \wedge {\varphi^{\beta} \subst{\scalebox{0.865}{\replace{z}{x},\replace{t}{y}}}}}}
          {\infer[(\text{Cut})]
            {\sequent{\Sigma(z)}{\exists z, m \suchthat m < n \wedge (s = z \vee A(m, s, z)) \wedge {\varphi^{\beta} \subst{\scalebox{0.865}{\replace{z}{x},\replace{t}{y}}}}}}
            {\infer*{\sequent{}{z = 0 \vee \exists z' \suchthat z = \suc{z'}}}{\dag}
              &
             \infer[(\text{$\vee$L})]
              {\sequent{z = 0 \vee \exists z' \suchthat z = \suc{z'}, \Sigma(z)}{\psi}}
              {\infer*{\sequent{z = 0, \Sigma(z)}{\psi}}{{\subref*{fig:CRTC+A-Beta-in-CA:TraceAwareDerivation:OneStep}}}
                &
               \infer*{\sequent{\exists z' \suchthat z = \suc{z'}, \Sigma(z)}{\psi}}{{\subref*{fig:CRTC+A-Beta-in-CA:TraceAwareDerivation:MultiStep}}}
              }
            }
          }
      \end{gather*}
    \end{minipage}}
    \caption{The core subderivation of the simulation of Rule \eqref{eq:RTC-cyc} in {\ca}.}
    \label{fig:CRTC+A-Beta-in-CA:TraceAwareDerivation:Subderivation}  
  \end{figure*}

  Take an instance of Rule \eqref{eq:RTC-cyc} with contexts $\Gamma$ and $\Delta$, and active formulas $(\rtcformula{x}{y}{\varphi})(s, t)$ and $(\rtcformula{x}{y}{\varphi})(s, z)$ in the conclusion and right-hand premise, respectively. For terms $r$, $s$, and $t$, let:
  \begin{itemize}[nosep]
    \item
    $\vartheta(r)$ abbreviate $\exists v, w \suchthat B(c,r,v) \wedge B(c,\suc{r},w) \wedge {\varphi^{\beta} \subst{\scalebox{0.865}{\replace{v}{x},\replace{w}{y}}}}$; and 
    \item
    $A(r, s, t)$ abbreviate $\exists z, c \suchthat r = \suc{z} \wedge B(c,0,s) \wedge B(c,\suc{z},t) \wedge \forall u \leq z \holdsthat \vartheta(u)$.
  \end{itemize}
  Additionally, let $\Sigma(r)$ and $\Pi(r)$ abbreviate the follow sequences of formulas:
  \begin{gather*}
    n = \suc{r}, B(c, 0, s), B(c, \suc{r}, t), \forall u \leq r \holdsthat \vartheta(u)
      \\
    n = \suc{\suc{r}}, B(c, 0, s), B(c, \suc{\suc{r}}, t), B(c,\suc{r},v), B(c,\suc{\suc{r}},w), {\varphi^{\beta} \subst{\scalebox{0.865}{\replace{v}{x},\replace{w}{y}}}}, \forall u \leq r \holdsthat \vartheta(u)
  \end{gather*}
  Moreover, note the following.
  \begin{enumerate}[nosep, wide, label={\roman*)}]
    \item
    We can easily derive $\sequent{}{z = 0 \vee \exists z' \suchthat z = \suc{z'}}$ using standard first-order rules and the axioms of {\ca}; we refer to this derivation using $\dag$.
    \item
    Our use of the notation $\forall u \leq t \holdsthat \gamma$ technically abbreviates the {\ca} formula $\forall u \holdsthat (u = t \vee u < t) \rightarrow \gamma$, and so we may straightforwardly derive both %the sequents 
    $\sequent{\forall u \leq 0 \holdsthat \gamma(u)}{\gamma(0)}$ and $\sequent{\forall u \leq \suc{t} \holdsthat \gamma(u)}{\gamma(\suc{t}) \wedge \forall u \leq t \holdsthat \gamma(u)}$;
    for brevity, we refer to an instance of the (Cut) rule that applies these sequents using the labels $(\leq_{0})$ and $(\leq_{\suc{}})$, respectively.
    \item
    Recall that, since the formula $B$ captures a $\beta$-function, we may also derive $\sequent{B(r, s, t), B(r, s, u)}{t = u}$; we abbreviate instances of (Cut) that apply an instance of this sequent using the label ($B$).
  \end{enumerate}
  
  Using these elements, \cref{fig:CRTC+A-Beta-in-CA:TraceAwareDerivation:Subderivation} shows a derivation of the following sequent, in which we have abbreviated the antecedent formula by $\psi$: 
  \begin{gather*}
    \sequent{A(n, s, t)}{\exists z, m \suchthat m < n \wedge (s = z \vee A(m, s, z)) \wedge {\varphi^{\beta} \subst{\scalebox{0.865}{\replace{z}{x},\replace{t}{y}}}}}
  \end{gather*}
  Then, using \cref{fig:CRTC+A-Beta-in-CA:TraceAwareDerivation:Subderivation} as a subderivation, we derive the $\bar{\beta}$-translation of Rule \eqref{eq:RTC-cyc} in {\ca} as shown in \cref{fig:CRTC+A-Beta-in-CA:TraceAwareDerivation}. Note that for any sequence of formulas $\Sigma$, we can straightforwardly derive $\sequent{\Gamma, \Sigma^{\beta}}{\Delta}$ from $\sequent{\Gamma, \Sigma^{\bar{\beta}}}{\Delta}$ by first introducing existential quantifiers for the free variable parameters in $\Sigma^{\bar{\beta}}$ and then eliminating the terms $n < \suc{z}$ with cuts. We abbreviate such a derivation using the label $(\bar{\beta})$. Note that this admits non-progressing traces for all the free variable parameters in $\Gamma$.
  The crucial feature of this derivation is that there is a {\ca} trace from the free variable parameter $n$ in the conclusion to $m$ in the right-hand premise, which progresses at the sequent containing the boxed formula $m < n$. Also, since the context $\Gamma^{\bar{\beta}}$ is preserved along the paths to both the left and right premises, all non-progressing traces are simulated as well.
  
  \begin{figure*}[!t]
    \resizebox{\textwidth}{!}{\begin{minipage}{1.18\textwidth}
      \begin{gather*}
        \infer[{(\text{$\vee$L})}]
          {\sequent{\Gamma^{\bar{\beta}}, s = t \vee A(n, s, t)}{\Delta^\beta}}
          {\infer[(\text{$=$L})]
            {\sequent{\Gamma^{\bar{\beta}}, s = t}{\Delta^\beta}}
            {\sequent
              {(\Gamma')^{\bar{\beta}} \subst{\replace{s}{u}, \replace{t}{w}}}
              {(\Delta')^\beta \subst{\replace{s}{u}, \replace{t}{w}}}}
              &
           \multiput(60,0)(0,9){5}{.}
           \raisebox{1.5cm}{
           \hbox to 8cm {
           \kern -2.5cm
           \infer[(\text{Cut})]
            {\sequent{\Gamma^{\bar{\beta}}, A(n, s, t)}{\Delta^\beta}}
            {\multiput(14,0)(0,10){8}{.}
             \raisebox{8em}{
             \hbox to 2em {
               \kern-12em
               \infer*{\sequent{A(n,s,t)}{\exists z, m \suchthat m < n \wedge (s = z \vee A(m, s, z)) \wedge {\varphi^{\beta} \subst{\scalebox{0.865}{\replace{z}{x},\replace{t}{y}}}}}}{\text{\small \cref{fig:CRTC+A-Beta-in-CA:TraceAwareDerivation:Subderivation}}}
             }}
              &
             \infer[(\text{$\exists$L/$\wedge$L})]
               {\sequent{\Gamma^{\bar{\beta}}, \exists z, m \suchthat m < n \wedge (s = z \vee A(m, s, z)) \wedge {\varphi^{\beta} \subst{\scalebox{0.865}{\replace{z}{x},\replace{t}{y}}}}}{\Delta^\beta}}
               {\infer[(\text{WL})]
                 {\sequent{\Gamma^{\bar{\beta}}, \fbox{$m < n$}, (s = z \vee A(m, s, z)), {\varphi^{\beta} \subst{\scalebox{0.865}{\replace{z}{x},\replace{t}{y}}}}}{\Delta^\beta}}
                 {\infer[({\bar{\beta}})]
                    {\sequent{\Gamma^{\bar{\beta}}, (s = z \vee A(m, s, z)), {\varphi^{\beta} \subst{\scalebox{0.865}{\replace{z}{x},\replace{t}{y}}}}}{\Delta^\beta}}
                    {\sequent{\Gamma^{\bar{\beta}}, (s = z \vee A(m, s, z)), {\varphi^{\bar{\beta}} \subst{\scalebox{0.865}{\replace{z}{x},\replace{t}{y}}}}}{\Delta^\beta}}
                 }
               }
            }}}
          }
      \end{gather*}
    \end{minipage}}
    \caption{A derivation schema simulating Rule \eqref{eq:RTC-cyc} in {\ca}.}%; here $C$, $D$, $E$, and $F$ stand for, respectively, $s = z$, $\subst{{\scriptstyle\replace{z}{x}},{\scriptstyle\!\replace{t}{y}}}$, $v < w$, and $A_v(s, z)$.}
    \label{fig:CRTC+A-Beta-in-CA:TraceAwareDerivation}
  \end{figure*}

  Now, using the derivations of the local $\bar{\beta}$-translations of the inference rules, from a {\cortc\plusArith} pre-proof, we can build a {\ca} pre-proof with the same global structure. For each bud in the resulting {\ca} pre-proof, we first apply an instance of the substitution rule that substitutes each free variable parameter of its companion with the free variable parameter of its corresponding $\bar{\beta}$-translation instance. Notice that this is possible, since the parameter variable is unique for the $\bar{\beta}$-translation of each $\rtcOperator$ sub-formula. We can then form a cycle in the {\ca} pre-proof. 
  
  Since the {\cortc} traces for each rule are simulated by the {\ca} derived rules, for each trace following a (finite or infinite) path in the {\cortc\plusArith} pre-proof there is a trace following the corresponding path in the {\ca} pre-proof containing a progression point for each progression point in the {\cortc\plusArith} trace. From this it follows that if the {\cortc\plusArith} pre-proof satisfies the ({\ortc}) global trace condition, then its translation satifies the {\ca} global trace condition.
  Finally, we derive $\sequent{\Gamma^{\beta}}{\Delta^{\beta}}$ from $\sequent{\Gamma^{\bar{\beta}}}{\Delta^{\beta}}$ as described above.
  \qed
\end{proof}

\begin{restatable}{theorem}{CRTCEquivCA}
  \label{lem:CA-CRTC}
 $\vdash_{\cortc\plusArith} \sequent{\Gamma}{\Delta}$ iff ${} \vdash_{\ca} \sequent{\Gamma^{\beta}}{\Delta^{\beta}}$.   
 %   In particular, for $\Gamma$ and $\Delta$ in the language of {\PA}, ${} \vdash_{\cortc\plusArith} \sequent{\Gamma}{\Delta}$ iff ${} \vdash_{\ca} \sequent{\Gamma}{\Delta}$.
\end{restatable}
\begin{proof}
  By \Cref{cor:CA-Beta-in-CRTC+A,lem:CRTC+A-Beta-in-CA}.
  \qed
\end{proof}

These results allow us to show an equivalence between the finitary and cyclic systems for {\TC} with arithmetic.

\begin{restatable}{theorem}{ArithmeticEquivalence}
  \label{thm:ArithmeticEquiv}
  $\rtc\plusArith$ and $\cortc\plusArith$ are equivalent. %(cf.~\cite{Simpson2017}).
\end{restatable}
\begin{proof}
  The fact that $\rtc\plusArith \subseteq \cortc\plusArith$ follows immediately from \Cref{easyinclusion}. For the converse, suppose $\sequent{\Gamma}{\Delta}$ is provable in $\cortc\plusArith$. By \Cref{lem:CA-CRTC} we get that $ \vdash_{\ca} \sequent{\Gamma^{\beta}}{\Delta^{\beta}}$. Using the equivalence between {\ca} and {\pa}, we obtain $ \vdash_{\pa} \sequent{\Gamma^{\beta}}{\Delta^{\beta}}$.\lcnote{can we use Simpson's result for languages extending PA's? NOT SURE} Then we conclude using \Cref{lem:PA-RTC}\eqref{lem:PA-RTC:sequent}.
  \qed
%  The claim follows from the above two lemmas together with the result of \cite{Simpson2017} which states that {\ca} is equivalent to {\pa}.
\end{proof}

%\lcnote{It seems like we have a stronger theorem here. I think it is true that the result of $\pa = \ca$ holds in any extension of those systems. That is, for any expansion of the language of {\PA}, and any set $S$ additional axioms, ${\pa}{+}S = {\ca}{+}S$. Now the $\beta$-translation works for such extensions as well so we could get that $\rtc\plusArith{+}S$ and $\cortc\plusArith{+}S$ are equiv. This gets us closer (maybe exactly) to the extended result in \cite{Berardi2017b}.}

Note that the result above can easily be extended to show that adding the same set of additional axioms to both $\rtc\plusArith$ and $\cortc\plusArith$ results in equivalent systems.
Also note that in the systems with pairs, to embed arithmetics there is no need to explicitly include addition and its axioms. Thus, by only including the signature $\{ 0,\fn{s} \}$ and the corresponding axioms for it we can obtain that $\rtcPairs\plusArith$ and $\cortcPairs\plusArith$ are equivalent. 

In \cite{Berardi2017b}, the equivalence result of \cite{Simpson2017} was improved to show it holds for any set of inductive predicates containing the natural number predicate $\pred{N}$. %Building on Simpson's result we obtain the equivalence between the $RTC$ systems for languages which \emph{include} {\PA}'s. 
On the one hand, our result goes beyond that of \cite{Berardi2017b} as it shows the equivalence for systems with a richer notion of inductive definition, due to the expressiveness of {\TC}. On the other hand, {\TC} does not support restricting the set of inductive predicates, i.e.~the {\rtcOperator} operator may operate on any formula in the language. To obtain a finer result which corresponds to that of \cite{Berardi2017b} we need to further explore the transformations between proofs in the two systems. This is left for future work. 

\subsubsection{The General Case}
\label{hydra}

As mentioned, the general equivalence conjecture between $\lkid$ and $\colkid$ was refuted in \cite{Berardi2017}, by providing a concrete example of a statement which is provable in the cyclic system but not in the explicit one. The statement (called 2-Hydra) involves a predicate encoding a binary version of the `hydra' induction scheme for natural numbers given in \cite{KirbyParis82}, and expresses that every pair of natural numbers is related by the predicate.\footnote{In fact, the falsifying Henkin model constructed in \cite{Berardi2017} also satisfies the `$0$-axiom' ($\forall x . 0 \neq \suc{x}$), and the `$\suc{}$-axiom' ($\forall x, y . \suc{x} = \suc{y} \rightarrow x = y$) stipulating injectivity of the successor function, and so the actual counter-example to equivalence is the sequent: $\sequent{\text{$(0,\suc{})$-axioms}}{\text{2-Hydra}}$.} 
However, a careful examination of this counter-example reveals that it only refutes a strong form of the conjecture, according to which both systems are based on the same set of productions. In fact, already in \cite{Berardi2017} it is shown that if the explicit system is extended by another inductive predicate, namely one expressing the $\leq$ relation, then the 2-Hydra counter-example becomes provable.
Therefore, the less strict formulation of the question, namely whether for any proof in $\colkid_\phi$ there is a proof in $\lkid_{\phi'}$ for some $\phi' \supseteq \phi$, has not yet been resolved.
Notice that in {\TC} the equivalence question is of this weaker variety, since the $\rtcOperator$ operator `generates' all inductive definitions at once. That is, there is no \emph{a priori} restriction on the inductive predicates one is allowed to use. Indeed, the 2-Hydra counter-example from \cite{Berardi2017} can be expressed in $\rtclogic$ and proved in $\cortc$. However, this does not produce a counter-example for {\TC} since it is also provable in $\rtc$, due to the fact that $s \leq t$ is definable via the {\rtcOperator} formula $(\rtcformula{w}{u}{\suc{w} = u})(s, t)$.

Despite our best efforts, we have not yet managed to settle this question, which appears to be harder to resolve in the {\TC} setting.
One possible approach to solving it is the semantical one, i.e.~exploiting the fact that the explicit system is known to be sound w.r.t. Henkin semantics. This is what was done in \cite{Berardi2017}. Thus, to show strict inclusion one could construct an alternative statement that is provable in $\cortc$ whilst also demonstrating a Henkin model for {\TC} that is not a model of the statement. However, constructing a {\TC} Henkin model appears to be non-trivial, due to its rich inductive power. In particular, it is not at all clear whether the structure that underpins the {\lkid} counter-model for 2-Hydra admits a Henkin model for TC. Alternatively, to prove equivalence, one could show that $\cortc$ is also sound w.r.t. Henkin semantics. Here, again, proving this does not seem to be straightforward.

In our setting, there is also the question of the inclusion of {\cortc} in {\ncortc}, which amounts to the question of whether overlapping cycles can be eliminated.
Moreover, we can ask if {\ncortc} is included in {\rtc}, independently of whether this also holds for {\cortc}. Again, the semantic approach described above may prove fruitful in answering these questions.

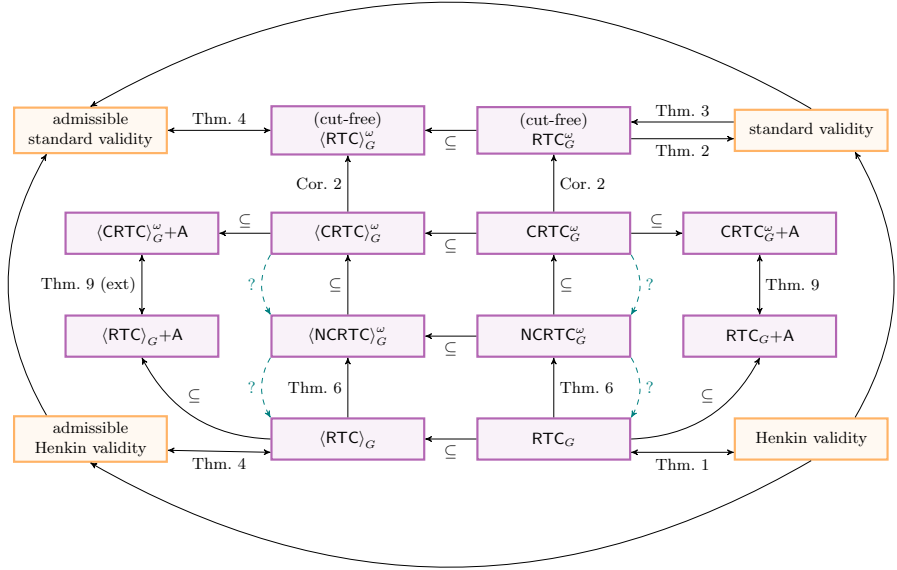
\begin{figure}[t]
  \begin{center}
  \vspace{-3em}
  \resizebox{\textwidth}{!}{
  \begin{tikzpicture}[
      roundnode/.style={circle, draw=violet!60, fill=violet!5, very thick, minimum size=10mm}, 
      squarednode/.style={rectangle, draw=violet!60, fill=violet!5, very thick, minimum size=8mm}, 
      semnode/.style={rectangle, draw=orange!60, fill=orange!5, very thick, minimum size=8mm}, 
      lkidnode/.style={rectangle, draw=green!60, fill=green!5, very thick, minimum size=8mm},
      ,>=stealth']
    
    \node[semnode]     at ( 5, 2) (standard)      {\parbox{2.75cm}{\centering{standard validity}}}; 
    \node[semnode]     at (-9, 2) (adstandard)    {\parbox{2.75cm}{\centering{admissible standard validity}}}; 
    \node[semnode]     at ( 5,-4) (henkin)        {\parbox{2.75cm}{\centering{Henkin validity}}};
    \node[semnode]     at (-9,-4) (adhenkin)      {\parbox{2.75cm}{\centering{admissible{\\}Henkin validity}}};
    
%    \node[lkidnode]    at (-4  , 4) (olkid)      {\parbox{2.75cm}{\centering{provability in \olkid}}}; 
%    \node[lkidnode]    at (-4  ,-6) (lkid)       {\parbox{2.75cm}{\centering{provability in \lkid}}}; 

    \node[squarednode] at ( 0, 2) (ortc)          {\parbox{2.75cm}{\centering{(cut-free)\\\ortc}}}; 
    \node[squarednode] at (-4, 2) (ortc-p)        {\parbox{2.75cm}{\centering{(cut-free)\\\ortcPairs}}}; 
    \node[squarednode] at (-4, 0) (cortc-p)       {\parbox{2.75cm}{\centering{\cortcPairs}}};
    \node[squarednode] at ( 0, 0) (cortc)         {\parbox{2.75cm}{\centering{\cortc}}};
    \node[squarednode] at (-4,-2) (ncortc-p)      {\parbox{2.75cm}{\centering{\ncortcPairs}}};
    \node[squarednode] at ( 0,-2) (ncortc)        {\parbox{2.75cm}{\centering{\ncortc}}};
    \node[squarednode] at ( 0,-4) (rtc)           {\parbox{2.75cm}{\centering{\rtc}}};
    \node[squarednode] at (-4,-4) (rtc-p)         {\parbox{2.75cm}{\centering{\rtcPairs}}};
    
    \node[squarednode] at (-8, 0) (arith-cortc-p) {\parbox{2.75cm}{\centering{\cortcPairs\plusArith}}};
    \node[squarednode] at ( 4, 0) (arith-cortc)   {\parbox{2.75cm}{\centering{\cortc\plusArith}}};
    \node[squarednode] at (-8,-2) (arith-rtc-p)   {\parbox{2.75cm}{\centering{\rtcPairs\plusArith}}};
    \node[squarednode] at ( 4,-2) (arith-rtc)     {\parbox{2.75cm}{\centering{\rtc\plusArith}}};

    % Relationships between the semantics
    \draw
      ([xshift=-2em] henkin.north east) 
        edge[bend right,->] 
       ([xshift=-2em] standard.south east);
    \draw
      ([xshift=2em] adhenkin.north west) 
        edge[bend left,->]
      ([xshift=2em] adstandard.south west);
    \draw
      (standard.north)
        edge[bend right,->]
      (adstandard.north);
    \draw
      (henkin.south)
        edge[bend left,->]
      (adhenkin.south);
    
    % Soundness and Completeness results
    \draw
      ([yshift=-0.5em] ortc.east) 
        edge[->] node[below, midway] {\small{Thm.~\ref{sound}}}
      ([yshift=-0.5em] standard.west);
    \draw
      ([yshift=0.5em] standard.west)
        edge[->] node[above] {\small{Thm.~\ref{thm:ORTC:Cut-free-complete}}}
      ([yshift=0.5em] ortc.east);
    \draw
      ([yshift=0.5em] rtc.south east)
        edge[<->] node[below,midway] {\small{Thm.~\ref{rtc_comp}}}
      ([yshift=0.5em] henkin.south west);
    \draw[<->]
      (ortc-p) 
        edge node[above] {\small{Thm. \ref{rtcpairs_comp}}}
      (adstandard);
    \draw
      ([yshift=0.5em] rtc-p.south west) 
        edge[<->] node[below] {\small{Thm. \ref{rtcpairs_comp}}}
      ([yshift= 0.75em] adhenkin.south east);

    % Relationships with LKID
%    \draw[<->, dashed, color=teal]
%      (lkid) edge node[right] {\small\centering {Sec. \ref{sec:lkid}}} (rtc-p);
%    \draw[<->, dashed, color=teal]
%      (olkid) edge node[right] {\small\centering {Sec. \ref{sec:lkid}}} (ortc-p);

    % Relationships between the different RTC systems
%    \draw[->]
%      (rtc) -- (cortc) node[right,midway] {\small Thm. \ref{easyinclusion} };
%    \draw[->]
%      (rtc-p) -- (cortc-p) node[right,midway] {\small Thm. \ref{easyinclusion} };
    \draw[->]
      (cortc) -- (cortc-p) 
        node[below,midway] {\small{$\subseteq$}};
    \draw[->]
      (ortc) -- (ortc-p)
        node[below,midway] {\small{$\subseteq$}};
    \draw[->]
      (ncortc) -- (ncortc-p)
        node[below,midway] {\small{$\subseteq$}};
    \draw[->]
      (rtc) -- (rtc-p)
        node[below,midway] {\small{$\subseteq$}};
    \draw[->]
      (cortc) -- (ortc)
        node[right,midway] {\small{Cor.~\ref{thm:ORTC:Cut-elim}}};
    \draw[->]
      (cortc-p) -- (ortc-p)
        node[left,midway] {\small{Cor.~\ref{thm:ORTC:Cut-elim}}};
    \draw[->]
      (ncortc) -- (cortc)
        node[right,midway] {\small $\subseteq$};
    \draw[->]
      (ncortc-p) -- (cortc-p)
        node[left,midway] {\small $\subseteq$};
    \draw[->]
      (rtc) -- (ncortc)
        node[right,midway] {\small{Thm.~\ref{normal}}};
    \draw[->]
      (rtc-p) -- (ncortc-p)
        node[left,midway] {\small{Thm.~\ref{normal}}};
      
    % Unknown Relationships
    \draw[dashed, color=teal]
      (cortc-p.south west)
        edge[bend right,->] node[left] {\small{?}}
      (ncortc-p.north west);
    \draw[dashed, color=teal]
      (cortc.south east)
        edge[bend left,->] node[right] {\small{?}}
      (ncortc.north east);
    \draw[dashed, color=teal]
      (ncortc-p.south west)
        edge[bend right,->] node[left] {\small{?}}
      (rtc-p.north west);
    \draw[dashed, color=teal]
      (ncortc.south east)
        edge[bend left,->] node[right] {\small{?}}
      (rtc.north east);
      
    % Arithmetical Systems
    \draw[<->]
      (arith-rtc) -- (arith-cortc)
        node[right,midway] {\small{Thm.~\ref{thm:ArithmeticEquiv}}};
    \draw[,<->]
      (arith-rtc-p.north) -- (arith-cortc-p.south)
        node[left,midway] {\small{Thm.~\ref{thm:ArithmeticEquiv} (ext)}};
    \draw[->]
      (cortc-p) -- (arith-cortc-p)
        node[above,midway] {\small{$\subseteq$}};
    \draw[->]
      (cortc) -- (arith-cortc)
        node[above,midway] {\small{$\subseteq$}};
    \draw
      (rtc.east) 
        edge[bend right,->] node[yshift=0.5em ,above] {\small{$\subseteq$}}
      (arith-rtc.south);
    \draw
      (rtc-p.west)
        edge[bend left,->] node[yshift=0.5em ,above] {\small{$\subseteq$}}
      (arith-rtc-p.south);
  \end{tikzpicture}}
  \vspace{-3em}
  \end{center}
  \caption{Diagrammatic Summary of our Results.}
  \label{fig:proof-system-relationships}
\end{figure}

\section{Conclusions and Future Work}
\label{sec:FutureWork}

We developed a natural infinitary proof system for transitive closure logic which is cut-free complete for the standard semantics and subsumes the explicit system. We further explored its restriction to cyclic proofs which provides the basis for an effective system for automating inductive reasoning. In particular, we syntactically identified a subset of cyclic proofs that is Henkin-complete.
%The main results of this paper are summarized in the figure below (adopted from \cite{brotherston2010sequent}). 
A summary of the proof systems we have studied in this paper, and their interrelationships, is shown in \Cref{fig:proof-system-relationships}. Where an edge between systems is labelled with an inclusion $\subseteq$, this signifies that a proof in the source system is already a proof in the destination system.

As mentioned in the introduction, as well as throughout the paper, this research was motivated by other work on systems of inductive definitions, particularly the {\lkid} framework of \cite{brotherston2010sequent}, its infinitary counterpart {\olkid}, and its cyclic subsystem {\colkid}. In terms of the expressive power of the underlying logic, {\TC} (assuming pairs) subsumes the inductive machinery underlying {\lkid}. This is because for any inductive predicate $P$ of {\lkid}, there is an {\rtclogic} formula $\psi$ such that for every standard admissible structure $M$ for {\rtclogic}, $P$ has the same interpretation as $\psi$ under $M$. This is due to Thm.~3 in \cite{AvronTC03} and the fact that the interpretation of $P$ must necessarily be a recursively enumerable set. As for the converse inclusion, for any positive $\rtclogic$ formula there is a production of a corresponding $\lkid$ inductive definition. %This is since productions can be seen as Horn clauses, thus they can capture disjunction and conjunction.
However, the {\rtcOperator} operator can also be applied on complex formulas (whereas {\lkid} productions only consider atomic predicates). This indicates that {\TC} might be more expressive. It was noted in \cite[p.~1180]{brotherston2010sequent} that complex formulas may be handled by stratifying the theory of {\lkid}, similar to \cite{Martin-Lof71}, but the issue of relative expressiveness of the resulting theory is not addressed. While we strongly believe it is the case that {\TC} is strictly more expressive than the logic of {\lkid}, proving so is left for future work. Also left for future research is establishing the comparative status of the corresponding formal proof systems.

In addition to the open question of the (in)equivalence of {\rtc} and {\cortc} in the general case, discussed in \Cref{cortc-rtc}, several other questions and directions for further study naturally arise from the work of this paper.
An obvious one would be to implement our cyclic proof system in order to investigate the practicalities of using {\TC} logic to support automated inductive reasoning.
More theoretically it is already clear that {\TC} logic, as a framework, diverges from existing systems for inductive reasoning (e.g.~{\lkid}) in interesting, non-trivial ways. The uniformity provided by the transitive closure operator may offer a way to better study the relationship between implicit and explicit induction, e.g.~in the form of cuts required in each system, or the relative complexity of proofs that each system admits.
Moreover, it seems likely that \emph{co}inductive reasoning can also be incorporated into the formal system. Determining whether, and to what extent, these are indeed the case is left for future work.

\bibliographystyle{splncs03}
\bibliography{cyclic_tc-TR}

\begin{thebibliography}{10}
\providecommand{\url}[1]{\texttt{#1}}
\providecommand{\urlprefix}{URL }

\bibitem{AvronTC03}
Avron, A.: {T}ransitive {C}losure and the {M}echanization of {M}athematics. In:
  Kamareddine, F.D. (ed.) Thirty Five Years of Automating Mathematics, Applied
  Logic Series, vol.~28, pp. 149--171. Springer, Netherlands (2003),
  doi:\url{10.1007/978-94-017-0253-9_7}

\bibitem{Berardi2017}
Berardi, S., Tatsuta, M.: {C}lassical {S}ystem of {M}artin-{L}{\"{o}}f's
  {I}nductive {D}efinitions {I}s {N}ot {E}quivalent to {C}yclic {P}roof
  {S}ystem. In: Proceedings of the 20\textsuperscript{th} International
  Conference on Foundations of Software Science and Computation Structures,
  {FOSSACS} 2017, Uppsala, Sweden, April 22--29, 2017. pp. 301--317. Springer
  Berlin Heidelberg, Berlin, Heidelberg (2017),
  doi:\url{10.1007/978-3-662-54458-7_18}

\bibitem{Berardi2017b}
Berardi, S., Tatsuta, M.: {E}quivalence of {I}nductive {D}efinitions and
  {C}yclic {P}roofs {U}nder {A}rithmetic. In: Proceedings of the
  32\textsuperscript{nd} Annual {ACM/IEEE} Symposium on Logic in Computer
  Science, {LICS} 2017, Reykjavik, Iceland, June 20--23, 2017. pp. 1--12
  (2017), doi:\url{10.1109/LICS.2017.8005114}

\bibitem{Brotherston07}
Brotherston, J.: {F}ormalised {I}nductive {R}easoning in the {L}ogic of
  {B}unched {I}mplications. In: Proceedings of Static Analysis,
  14\textsuperscript{th} International Symposium, {SAS} 2007, Kongens Lyngby,
  Denmark, August 22--24, 2007. pp. 87--103 (2007),
  doi:\url{10.1007/978-3-540-74061-2_6}

\bibitem{BrotherstonBC08}
Brotherston, J., Bornat, R., Calcagno, C.: {C}yclic {P}roofs of {P}rogram
  {T}ermination in {S}eparation {L}ogic. In: Proceedings of the
  35\textsuperscript{th} {ACM} {SIGPLAN-SIGACT} Symposium on Principles of
  Programming Languages, {POPL} 2008, San Francisco, California, USA, January
  7--12, 2008. pp. 101--112 (2008), doi:\url{10.1145/1328438.1328453}

\bibitem{brotherston2010sequent}
Brotherston, J., Simpson, A.: {S}equent {C}alculi for {I}nduction and
  {I}nfinite {D}escent. {J}ournal of {L}ogic and {C}omputation  21(6),
  1177--1216 (2010), doi:\url{10.1093/logcom/exq052}

\bibitem{buss1998handbook}
Buss, S.R.: {H}andbook of {P}roof {T}heory. Studies in Logic and the
  Foundations of Mathematics, Elsevier Science (1998)

\bibitem{CohenMSc}
Cohen, L.: {A}ncestral {L}ogic and {E}quivalent {S}ystems. Master's thesis,
  Tel-Aviv University, Israel (2010)

\bibitem{Cohen2017Henkin}
Cohen, L.: {C}ompleteness for {A}ncestral {L}ogic via a
  {C}omputationally-{M}eaningful {S}emantics. In: Proceedings of the
  26\textsuperscript{th} International Conference on Automated Reasoning with
  Analytic Tableaux and Related Methods, {TABLEAUX} 2017, Bras{\'{\i}}lia,
  Brazil, September 25--28, 2017. pp. 247--260 (2017),
  doi:\url{10.1007/978-3-319-66902-1_15}

\bibitem{Cohen2014AL}
Cohen, L., Avron, A.: {A}ncestral {L}ogic: {A} {P}roof {T}heoretical {S}tudy.
  In: et~al., U.K. (ed.) Logic, Language, Information, and Computation, Lecture
  Notes in Computer Science, vol. 8652, pp. 137--151. Springer (2014),
  doi:\url{10.1007/978-3-662-44145-9_10}

\bibitem{cohen2015middle}
Cohen, L., Avron, A.: {T}he {M}iddle {G}round--{A}ncestral {L}ogic. Synthese
  pp. 1--23 (2015), doi:\url{10.1007/s11229-015-0784-3}

\bibitem{CohenRowe18}
Cohen, L., Rowe, R.N.S.: {U}niform {I}nductive {R}easoning in {T}ransitive
  {C}losure {L}ogic via {I}nfinite {D}escent. In: Proceedings of the
  27\textsuperscript{th} {EACSL} Annual Conference on Computer Science Logic,
  {CSL} 2018, September 4--7, 2018, Birmingham, UK. pp. 16:1--16:17 (2018),
  doi:\url{10.4230/LIPIcs.CSL.2018.16}

\bibitem{CookReckhow79}
Cook, S.A., Reckhow, R.A.: {T}he {R}elative {E}fficiency of {P}ropositional
  {P}roof {S}ystems. The Journal of Symbolic Logic  44(1),  36--50 (1979)

\bibitem{Courcelle83}
Courcelle, B.: {F}undamental {P}roperties of {I}nfinite {T}rees. Theor. Comput.
  Sci.  25,  95--169 (1983), doi:\url{10.1016/0304-3975(83)90059-2}

\bibitem{DasP17}
Das, A., Pous, D.: {A} {C}ut-{F}ree {C}yclic {P}roof {S}ystem for {K}leene
  {A}lgebra. In: Proceedings of the 26\textsuperscript{th} International
  Conference Automated Reasoning with Analytic Tableaux and Related Methods,
  {TABLEAUX} 2017, Bras{\'{\i}}lia, Brazil, September 25--28, 2017. pp.
  261--277 (2017), doi:\url{10.1007/978-3-319-66902-1_16}

\bibitem{Gentzen1935}
Gentzen, G.: {U}ntersuchungen {\"u}ber das {L}ogische {S}chlie{\ss}en. {I}.
  Mathematische Zeitschrift  39(1),  176--210 (1935),
  doi:\url{10.1007/BF01201353}

\bibitem{Henkin50Completeness}
Henkin, L.: {C}ompleteness in the {T}heory of {T}ypes. Journal of Symbolic
  Logic  15(2),  81--91 (1950), \url{http://www.jstor.org/stable/2266967}

\bibitem{kashima2008general}
Kashima, R., Okamoto, K.: {G}eneral {M}odels and {C}ompleteness of
  {F}irst-order {M}odal $\mu$-calculus. Journal of Logic and Computation
  18(4),  497--507 (2008), doi:\url{10.1093/logcom/exm077}

\bibitem{KirbyParis82}
Kirby, L., Paris, J.: {A}ccessible {I}ndependence {R}esults for {P}eano
  {A}rithmetic. Bulletin of the London Mathematical Society  14(4),  285--293
  (1982), doi:\url{10.1112/blms/14.4.285}

\bibitem{Martin-Lof71}
Martin-L\"{o}f, P.: {Ha}uptsatz for the {I}ntuitionistic {T}heory of {I}terated
  {I}nductive {D}efinitions. In: Fenstad, J.E. (ed.) Proceedings of the Second
  Scandinavian Logic Symposium, {S}tudies in {L}ogic and the {F}oundations of
  {M}athematics, vol.~63, pp. 179--216. Elsevier (1971),
  doi:\url{10.1016/S0049-237X(08)70847-4}

\bibitem{martin1984intuitionistic}
Martin-L{\"o}f, P., Sambin, G.: {I}ntuitionistic {T}ype {T}heory, vol.~9.
  Bibliopolis Napoli (1984)

\bibitem{RoweBrotherston17}
Rowe, R.N.S., Brotherston, J.: {A}utomatic {C}yclic {T}ermination {P}roofs for
  {R}ecursive {P}rocedures in {S}eparation {L}ogic. In: Proceedings of the
  6\textsuperscript{th} {ACM} {SIGPLAN} Conference on Certified Programs and
  Proofs, {CPP} 2017, Paris, France, January 16--17, 2017. pp. 53--65 (2017),
  doi:\url{10.1145/3018610.3018623}

\bibitem{Santocanale2002}
Santocanale, L.: {A} {C}alculus of {C}ircular {P}roofs and {I}ts {C}ategorical
  {S}emantics. In: Nielsen, M., Engberg, U. (eds.) Proceedings of the
  5\textsuperscript{th} International Conference on Foundations of Software
  Science and Computation Structures, FOSSACS 2002, Grenoble, France, April
  8--12, 2002. pp. 357--371. Springer Berlin Heidelberg, Berlin, Heidelberg
  (2002), doi:\url{10.1007/3-540-45931-6_25}

\bibitem{Simpson2017}
Simpson, A.: {C}yclic {A}rithmetic {I}s {E}quivalent to {P}eano {A}rithmetic.
  In: Proceedings of the 20\textsuperscript{th} International Conference on
  Foundations of Software Science and Computation Structures, {FOSSACS} 2017,
  Uppsala, Sweden, April 22--29, 2017. pp. 283--300 (2017),
  doi:\url{10.1007/978-3-662-54458-7_17}

\bibitem{Sprenger2003}
Sprenger, C., Dam, M.: {O}n the {S}tructure of {I}nductive {R}easoning:
  {C}ircular and {T}ree-{S}haped {P}roofs in the $\mu${C}alculus. In: Gordon,
  A.D. (ed.) Proceedings of the 6\textsuperscript{th} International Conference
  on Foundations of Software Science and Computation Structures, FOSSACS 2003,
  Warsaw, Poland, April 7--11, 2003. pp. 425--440. Springer Berlin Heidelberg,
  Berlin, Heidelberg (2003), doi:\url{10.1007/3-540-36576-1_27}

\bibitem{takeuti2013proof}
Takeuti, G.: {P}roof {T}heory. Courier Dover Publications (1987)

\bibitem{TellezBrotherston17}
Tellez, G., Brotherston, J.: {A}utomatically {V}erifying {T}emporal
  {P}roperties of {P}ointer {P}rograms with {C}yclic {P}roof. In: Proceedings
  of the 26\textsuperscript{th} International Conference on Automated
  Deduction, {CADE} 26, Gothenburg, Sweden, August 6--11, 2017. pp. 491--508
  (2017), doi:\url{10.1007/978-3-319-63046-5_30}

\end{thebibliography}

\end{document}